\documentclass{article}
\usepackage{amsmath,amssymb,graphicx,bbm,theorem,ifthen,pst-all,placeins}

\newcommand{\lp}[1]{\ifthenelse{\equal{#1}{0}}{(}{}\ifthenelse{\equal{#1}{1}}{\bigl(}{}\ifthenelse{\equal{#1}{2}}{\Bigl(}{}\ifthenelse{\equal{#1}{3}}{\biggl(}{}\ifthenelse{\equal{#1}{4}}{\Biggl(}{}\ifthenelse{\equal{#1}{5}}{\Biggl(}{}}
\newcommand{\rp}[1]{\ifthenelse{\equal{#1}{0}}{)}{}\ifthenelse{\equal{#1}{1}}{\bigr)}{}\ifthenelse{\equal{#1}{2}}{\Bigr)}{}\ifthenelse{\equal{#1}{3}}{\biggr)}{}\ifthenelse{\equal{#1}{4}}{\Biggr)}{}\ifthenelse{\equal{#1}{5}}{\Biggr)}{}}
\newcommand{\lbc}[1]{\ifthenelse{\equal{#1}{0}}{\{}{}\ifthenelse{\equal{#1}{1}}{\bigl\{}{}\ifthenelse{\equal{#1}{2}}{\Bigl\{}{}\ifthenelse{\equal{#1}{3}}{\biggl\{}{}\ifthenelse{\equal{#1}{4}}{\Biggl\{}{}\ifthenelse{\equal{#1}{5}}{\Biggl\{}{}}
\newcommand{\rbc}[1]{\ifthenelse{\equal{#1}{0}}{\}}{}\ifthenelse{\equal{#1}{1}}{\bigr\}}{}\ifthenelse{\equal{#1}{2}}{\Bigr\}}{}\ifthenelse{\equal{#1}{3}}{\biggr\}}{}\ifthenelse{\equal{#1}{4}}{\Biggr\}}{}\ifthenelse{\equal{#1}{5}}{\Biggr\}}{}}
\newcommand{\lba}[1]{\ifthenelse{\equal{#1}{0}}{\langle}{}\ifthenelse{\equal{#1}{1}}{\bigl\langle}{}\ifthenelse{\equal{#1}{2}}{\Bigl\langle}{}\ifthenelse{\equal{#1}{3}}{\biggl\langle}{}\ifthenelse{\equal{#1}{4}}{\Biggl\langle}{}\ifthenelse{\equal{#1}{5}}{\Biggl\langle}{}}
\newcommand{\rba}[1]{\ifthenelse{\equal{#1}{0}}{\rangle}{}\ifthenelse{\equal{#1}{1}}{\bigr\rangle}{}\ifthenelse{\equal{#1}{2}}{\Bigr\rangle}{}\ifthenelse{\equal{#1}{3}}{\biggr\rangle}{}\ifthenelse{\equal{#1}{4}}{\Biggr\rangle}{}\ifthenelse{\equal{#1}{5}}{\Biggr\rangle}{}}
\newcommand{\ve}[1]{\ifthenelse{\equal{#1}{0}}{|}{}\ifthenelse{\equal{#1}{1}}{\big|}{}\ifthenelse{\equal{#1}{2}}{\Big|}{}\ifthenelse{\equal{#1}{3}}{\bigg|}{}\ifthenelse{\equal{#1}{4}}{\Bigg|}{}\ifthenelse{\equal{#1}{5}}{\Bigg|}{}}
\newcommand{\lb}[1]{\ifthenelse{\equal{#1}{0}}{[}{}\ifthenelse{\equal{#1}{1}}{\bigl[}{}\ifthenelse{\equal{#1}{2}}{\Bigl[}{}\ifthenelse{\equal{#1}{3}}{\biggl[}{}\ifthenelse{\equal{#1}{4}}{\Biggl[}{}\ifthenelse{\equal{#1}{5}}{\Biggl[}{}}
\newcommand{\rb}[1]{\ifthenelse{\equal{#1}{0}}{]}{}\ifthenelse{\equal{#1}{1}}{\bigr]}{}\ifthenelse{\equal{#1}{2}}{\Bigr]}{}\ifthenelse{\equal{#1}{3}}{\biggr]}{}\ifthenelse{\equal{#1}{4}}{\Biggr]}{}\ifthenelse{\equal{#1}{5}}{\Biggr]}{}}
\newcommand{\srp}[3]{\ifthenelse{\equal{#1}{0}}{)^{#2}_{#3}}{}\ifthenelse{\equal{#1}{1}}{\bigr)^{#2}_{#3}}{}\ifthenelse{\equal{#1}{2}}{\Bigr)^{#2}_{#3}}{}
\ifthenelse{\equal{#1}{3}}{\biggr)^{#2}_{#3}}{}\ifthenelse{\equal{#1}{4}}{\Biggr)^{#2}_{#3}}{}\ifthenelse{\equal{#1}{5}}{\Biggr)^{#2}_{#3}}{}}
\newcommand{\srb}[3]{\ifthenelse{\equal{#1}{0}}{]^{#2}_{#3}}{}\ifthenelse{\equal{#1}{1}}{\bigr]^{#2}_{#3}}{}\ifthenelse{\equal{#1}{2}}{\Bigr]^{#2}_{#3}}{}
\ifthenelse{\equal{#1}{3}}{\biggr]^{#2}_{#3}}{}\ifthenelse{\equal{#1}{4}}{\Biggr]^{#2}_{#3}}{}\ifthenelse{\equal{#1}{5}}{\Biggr]^{#2}_{#3}}{}}
\newcommand{\srbc}[3]{\ifthenelse{\equal{#1}{0}}{{\}}^{#2}_{#3}}{}\ifthenelse{\equal{#1}{1}}{\bigr{\}}^{#2}_{#3}}{}\ifthenelse{\equal{#1}{2}}{\Bigr{\}}^{#2}_{#3}}{}
\ifthenelse{\equal{#1}{3}}{\biggr{\}}^{#2}_{#3}}{}\ifthenelse{\equal{#1}{4}}{\Biggr{\}}^{#2}_{#3}}{}\ifthenelse{\equal{#1}{5}}{\Biggr{\}}^{#2}_{#3}}{}}
\newcommand{\srba}[3]{\ifthenelse{\equal{#1}{0}}{\rangle^{#2}_{#3}}{}\ifthenelse{\equal{#1}{1}}{\bigr\rangle^{#2}_{#3}}{}\ifthenelse{\equal{#1}{2}}{\Bigr\rangle^{#2}_{#3}}{}
\ifthenelse{\equal{#1}{3}}{\biggr\rangle^{#2}_{#3}}{}\ifthenelse{\equal{#1}{4}}{\Biggr\rangle^{#2}_{#3}}{}\ifthenelse{\equal{#1}{5}}{\Biggr\rangle^{#2}_{#3}}{}}
\newcommand{\sve}[3]{\ifthenelse{\equal{#1}{0}}{|^{#2}_{#3}}{}\ifthenelse{\equal{#1}{1}}{\bigr|^{#2}_{#3}}{}\ifthenelse{\equal{#1}{2}}{\Bigr|^{#2}_{#3}}{}
\ifthenelse{\equal{#1}{3}}{\biggr|^{#2}_{#3}}{}\ifthenelse{\equal{#1}{4}}{\Biggr|^{#2}_{#3}}{}\ifthenelse{\equal{#1}{5}}{\Biggr|^{#2}_{#3}}{}}
\newcommand{\im}{\mathcal{I}m}
\newcommand{\re}{\mathcal{R}e}
\newcommand{\sh}{\ensuremath{\tmop{sh}}}

\newcommand{\cylindreb}[1]{\ensuremath{\tilde{S}_{#1}}}
\newcommand{\paraboleb}[1]{\ensuremath{S_{#1}}}
\newcommand{\disqueb}[1]{\ensuremath{D_{#1}}}
\newcommand{\demidisqueb}[1]{\ensuremath{D_{#1}^+}}
\newcommand{\demiplanb}{\ensuremath{\mathbbm{C}^+}}
\newcommand{\nevb}[1]{\ensuremath{\tilde{D}_{#1}}}
\newcommand{\parabole}[1]{{\paraboleb{#1}}}
\newcommand{\disque}[1]{{\disqueb{#1}}}
\newcommand{\demidisque}[1]{{\demidisqueb{#1}}}
\newcommand{\demiplan}[1]{{\demiplanb}#1}
\newcommand{\nev}[1]{{\nevb{#1}}}
\newcommand{\cylindre}[1]{{\cylindreb{#1}}}
\newcommand{\AMSclass}[1]{{\textbf{A.M.S. subject classification:} #1}}
\newcommand{\assign}{:=}
\newcommand{\keywords}[1]{{\textbf{Keywords:} #1}}
\newcommand{\nonesep}{}
\newcommand{\tmdummy}{$\mbox{}$}
\newcommand{\tmmathbf}[1]{\ensuremath{\boldsymbol{#1}}}
\newcommand{\tmop}[1]{\ensuremath{\operatorname{#1}}}
\newcommand{\tmscript}[1]{\text{\scriptsize{$#1$}}}
\newcommand{\tmstrong}[1]{\textbf{#1}}
\newcommand{\tmtextbf}[1]{{\bfseries{#1}}}
\newcommand{\tmtextit}[1]{{\itshape{#1}}}
\newcommand{\tmtextmd}[1]{{\mdseries{#1}}}
\newcommand{\tmtextup}[1]{{\upshape{#1}}}
\newenvironment{itemizedot}{\begin{itemize} }{\end{itemize}}
\newenvironment{proof}{\noindent\textbf{Proof\ }}{\hspace*{\fill}$\Box$\medskip}

\numberwithin{equation}{section}  
\numberwithin{figure}{section}  
\newtheorem{theorem}{Theorem}[section]
\newtheorem{lemma}[theorem]{Lemma}
\newtheorem{proposition}[theorem]{Proposition}
\newtheorem{corollary}[theorem]{Corollary}
\newtheorem{definition}[theorem]{Definition}
\newtheorem{remark}[theorem]{Remark}
\begin{document}

\title{Borel summation of the small time expansion of the heat kernel. The
scalar potential case\thanks{This paper has been written using the GNU TEXMACS scientific text editor.}
\thanks{\keywords{heat kernel, quantum mechanics,
semi-classical, asymptotic expansion, Wigner-Kirkwood expansion, Borel
summation, Poisson formula, Wiener formula, Feynman formula}; \AMSclass{35K08,
30E15, 35C20, 81Q30}}}\author{Thierry Harg\'e}\date{January 21, 2013}\maketitle

\begin{abstract}
  Let $p_t$ be the heat kernel associated to the operator $- \Delta + V (x)$
  defined on $\mathbbm{R}^{\nu}$. We prove, under rather strong assumptions on
  $V$, that the small time expansion of $p_t$ is Borel summable. If $V$ is
  defined on the torus, we prove a Poisson formula.
\end{abstract}

\section{Introduction}

Let $\nu \in \mathbbm{N}^{\ast}$ and $V$ be a regular square matrix-valued
function on $\mathbbm{R}^{\nu}$. Denote $\partial^2_x \assign \partial^2_{x_1}
+ \cdots + \partial_{x_{\nu}}^2$ and $(x - y)^2 : = (x_1 - y_1)^2 + \cdots +
(x_{\nu} - y_{\nu})^2$ for $x \in \mathbbm{R}^{\nu}$and $y \in
\mathbbm{R}^{\nu}$. Let $p_t (x, y)$ be the heat kernel associated to the
operator $- \partial^2_x + V (x)$. Let $p_t^{\tmop{conj}} (x, y)$ be the
conjugate heat kernel defined by
\begin{equation}
  \label{intro0} \text{$p_t (x, y)$=$(4 \pi t)^{- \frac{\nu}{2}} \exp \lp{2} -
  \frac{(x - y)^2}{4 t} \rp{2} p_t^{\tmop{conj}} (x, y)$} .
\end{equation}
Then the Minakshisundaram-Pleijel asymptotic expansion holds:
\begin{equation}
  \label{intro1} \text{$p_t^{\tmop{conj}} (x, y) =$} \mathbbm{1} + a_1 (x, y)
  t + \cdots + a_{r - 1} (x, y) t^{r - 1} + t^r \mathcal{O}_{t \rightarrow
  0^+} (1) .
\end{equation}
Here $\mathbbm{1}$ denotes the identity matrix. The expansion in
(\ref{intro1}) is not convergent in general. A goal of this paper is to prove,
under rather strong assumptions on $V$, that this expansion is Borel summable
and that its Borel sum is equal to $p_t^{\tmop{conj}} (x, y)$ (see definition
\ref{definitionbarbara2.5}). Borel summability allows to recover
$p_t^{\tmop{conj}} (x, y)$ with the help of the knowledge of the coefficients
$a_1 (x, y), a_2 (x, y), \ldots$ For instance, if these coefficients vanish
then $\text{$p_t^{\tmop{conj}} (x, y) =$} \mathbbm{1}$.

Assume now that $V$ is defined on the torus $\mathbbm{(R} /
\mathbbm{Z})^{\nu}$ with values in a space of $d \times d$ Hermitian matrices.
Let $\lambda_1 \leqslant \lambda_2 \leqslant \cdots \leqslant \lambda_n
\leqslant \cdots, \lambda_n \rightarrow + \infty$ \ be the eigenvalues of the
operator$- \partial^2_x + V (x)$ acting on periodic $\mathbbm{C}^d$-valued
functions. The trace of the heat kernel has an asymptotic expansion
\begin{equation}
  \label{intro5} \sum_{n = 1}^{+ \infty} e^{- \lambda_n t} = (4 \pi t)^{-
  \frac{\nu}{2}} \lp{1} d + a_1 t + \cdots + a_{r - 1} t^{r - 1} + t^r
  \mathcal{O}_{t \rightarrow 0^+} (1) \rp{1} .
\end{equation}
We shall prove the Poisson formula: for $t \in \mathbbm{C}$, $\mathcal{R} et >
0$,
\begin{equation}
  \label{intro6} \sum_{n = 1}^{+ \infty} e^{- \lambda_n t} = (4 \pi t)^{-
  \frac{\nu}{2}} \sum_{q \in \mathbbm{Z}^{\nu}} e^{- \frac{q^2}{4 t}} u_q (t),
\end{equation}
\begin{equation}
  \label{intro6.5} \text{with \ \ } u_q (t) = d + a_{1, q} t + \cdots + a_{r -
  1, q} t^{r - 1} + \cdots, q \in \mathbbm{Z}^{\nu} .
\end{equation}
In (\ref{intro6.5}), each expansion is Borel summable and $u_q$ denotes the
Borel sum of such an expansion. Since $a_{1, 0} = a_1$, $a_{2, 0} = a_2$,
$\ldots$, the expansion in (\ref{intro5}) is Borel summable but (\ref{intro6})
shows us that the knowledge of the coefficients $a_1, a_2 \ldots$ does not
allow one to recover the left hand side of (\ref{intro5}) by Borel summation.

Let us now state more precisely the assumptions on $V$. Let $\alpha \in
\mathbbm{R}$ and let $\mu$ be a Borel measure on $\mathbbm{R}^{\nu}$ with
values in some complex finite dimensional space of square matrices, such that,
for some $\varepsilon > 0$
\begin{equation}
  \label{intro3} \int_{\mathbbm{R}^{\nu}} \exp (\varepsilon \xi^2) d| \mu |
  (\xi) < + \infty .
\end{equation}
We suppose that $V (x) = \alpha x^2 - c (x)${\footnote{The minus sign in front
of $c (x)$ is chosen for future convenience.}} where
\[ c (x) = \int_{\mathbbm{R}^{\nu}} \exp (ix \cdot \xi) d \mu (\xi) . \]
$c$ must be viewed as a perturbation. Assuming $c$ is complex valued instead
of matrix valued does not simplify the method and does not change the results.
In particular, (\ref{intro3}) implies that $V$ is analytic on
$\mathbbm{C}^{\nu}$ and $c$ is bounded on $\mathbbm{R}^{\nu}$. In fact, our
results hold if $\alpha x^2$ is replaced by an arbitrary quadratic form on
$\mathbbm{R}^{\nu}$. For the sake of simplicity, we choose not to write the
proofs in this case (see [Ha6] for a generalization of the deformation
formula). Note also that functions $c$ such that $c (x) : = \exp (i x_1)$ or \
$c (x) : = \exp (- x^2)$ satisfy our assumptions.

Quantum mechanics gives many examples of divergent expansions. We focus on
semi-classical expansions related to the Schr\"odinger equation since they
present a lot of similarities with the small time expansion of the heat
kernel. The semi-classical viewpoint allows to expand a quantum quantity in
terms of powers of $h$. The coefficients of this expansion can be viewed as
classical quantities. Giving a meaning to the sum of this expansion by using
only the coefficients allows one to recover the quantum quantity by means of
classical quantities [B-B, V1, V2]. The same interpretation holds for the
expansion of the heat kernel for small $t$ (the parameter $t$ may be viewed as
$\beta$, the inverse of the temperature). The heat kernel of an operator can
be viewed as a quantum quantity: for instance, its trace, if it exists, gives
the partition function of the spectrum of the operator. The coefficients of
the expansion have a classical interpretation: for instance in (\ref{intro0}),
the main term in the exponential is the square of the distance between $x$ and
$y$, which is a classical quantity. See also Remark 8.4 in [Ha3].

Recovering quantum quantities with the help of the coefficients of their
semi-classical expansion is a question considered by Voros [V1, V2] and
Delabaere, Dillinger, Pham [D-D-P] for the one dimensional Schr\"odinger
equation. Their use of Borel summation is not elementary as ours. Their
assumptions, which allows to consider tunnelling for instance, involve to deal
with ramified singularities in the Borel plane (Ecalle's alien calculus). In a
following paper [Ha8], we prove that the $h$ expansion of the partition
function of the Schr\"odinger operator is Borel summable in all dimensions but
with restrictive assumptions on the potential allowing a simple Borel
summation process.

Concerning the expansion of the heat kernel for small $t$, we are not aware of
references using Borel summation. However, Colin de Verdi\`ere [Co1, Co2]
gives a Poisson formula in the case of a smooth Riemannian manifold. The
setting of his work is much more general than ours, but the result is
asymptotic and does not give an exact formula. The exact formula on the torus
is a direct consequence of the result on $\mathbbm{R}^{\nu}$: it is just
putting together independent expansions (one more time, we do not need alien
calculus). The case of the torus is very simple but singular.

It is convenient to write the potential $c$ as the Fourier transform of a
Borel measure. This point of view is used by many authors [It, Ga, A-H],
working with a rigorous definition of Wiener and Feynman integrals. The
assumptions on $c$ and the method allow to consider in a natural way $p_t (x,
y)$ with $x, y \in \mathbbm{C}^{\nu}, t \in \mathbbm{C}, \re t \geqslant 0$
instead of $x, y \in \mathbbm{R}^{\nu}, t \in \mathbbm{R} \cup i\mathbbm{R}$.

Our proofs use a so-called deformation formula and a so-called deformation
matrix. In the free case (i.e $V (x) = \alpha x^2 - c (x)$ with $\alpha = 0$),
this formula can be found in [It, Formula (77)]. In this reference, the
deformation matrix is not given explicitly. One can also give a formula for
the solution of the heat or Schr\"odinger equation with an arbitrary initial
condition (of course no factorization occurs contrary to the heat kernel
case), see \ [Ga, A-H] and more precisely [Ga, Eq.(27)] and [A-H,
Eq.(3.12),(5.16)]. In these references, the deformation formula is a mean to
study Feynman integral; the complex viewpoint, which is straightforward in the
free case, is not considered. In a heuristic way, this formula is well known
[On] and can even be considered as a particular case of a more general one
[Ge-Ya]. In this setting, there is no reason to consider $c$ as the Fourier
transform of a Borel measure. See also [Fu-Os-Wi] and [Ha3]. We strongly
advise the reader to look at the known heuristic proof of this formula given
in the Appendix. This proof, which uses Wiener representation of the heat
kernel and Wick's theorem, gives an explanation of the shape of the formula
and an interpretation of a so-called deformation matrix (see also [It]) which
is important in our method. However, there is a serious drawback in thinking
of the deformation formula as a consequence of the existence of Wiener or
Feynman integrals, in particular from a rigorous point of view. The
deformation formula is elementary: it does not use sophisticated notions about
infinite-dimensional spaces. Another heuristic proof, avoiding Wiener
representation, can be found in [On]. Here Wick's theorem also plays a central
role. In Section 4.2, we give a rigorous proof of the deformation formula,
which avoids Wiener representation but does not explain in a satisfactory way
the shape of the formula. We call it deformation formula because we want to
emphasize its perturbative nature. It is therefore not surprising that the
assumptions (\ref{anna1}) (the same as those used for the mathematical
foundation of the Feynman integrals [A-H]) or (\ref{anna12}) on the function
$c$ are highly restrictive. For instance, the heat kernel associated to the
operator $\partial_x^2 - x^2$ can not be viewed as a perturbation of the heat
kernel associated to the operator $\partial_x^2$ by this formula.

This formula gives an explicit solution of the heat equation viewed as a
partial differential equation which is available for every $\alpha \in
\mathbbm{R}$ and small $t \in \mathbbm{C}$ with $\re t > 0$. For this purpose
some properties of the deformation matrix for these values of $t$ are needed
(Section 4.1). At last, let us say a few words about the unicity problem. If
$\alpha \leqslant 0$, the heat kernel is uniquely defined as the kernel of an
analytic semi-group [Pa]. In the case $\alpha > 0$, see [Ha7].

The paper is organized as follows. In Section 2, we present some notation and
recall some classical facts about Borel summation. In Section 3, we state the
main results and give their proofs in Sections 4.3 and 4.4.

\medskip
{\small{\tmtextit{Acknowlegement}: I would like to express my gratitude to
Vladimir Georgescu for his pertinent advices. He also suggested many
improvements in the redaction.}}

\section{Preliminaries}

For $z = |z|e^{i \theta} \in \mathbbm{C}$, $\theta \in] - \pi, \pi]$, we
denote $z^{1 / 2} \assign |z|^{1 / 2} e^{i \theta / 2}$. Let $T > 0$. \ Let
$\demiplan{\assign \{z \in \mathbbm{C} | \re (z) > 0\}}$, $\disque{T} \assign
\{z \in \mathbbm{C} | |z| < T\}$, $\demidisque{T} \assign \disque{T} \cap
\demiplan{}$ and $\nev{T} \assign \lbc{1} z \in \mathbbm{C} | \re (
\frac{1}{z}) > \frac{1}{T} \rbc{1}$. $\nev{T}$ is the open disk of center
$\frac{T}{2}$ and radius $\frac{T}{2}$.

\begin{figure}[h]
 \caption{\label{cercle}}
 \centering
  \begin{pspicture}(-3,-2.4)(3,3.2)
  \psset{unit=1}
\psline[linewidth=0.1mm]{->}(-1,0)(3.5,0)
\psline[linewidth=0.1mm]{->}(0,-2.7)(0,2.7)
\parametricplot{0}{180}{2.5 t sin mul 2.5 t cos mul}
\parametricplot{0}{380}{1.25 1.25 t sin mul add 1.25 t cos mul}
\psline[linewidth=0.1mm]{<->}(1.25,0)(0.2,-0.7)
\psframe[linewidth=0mm,fillstyle=hlines](-1,-2.9)(0,2.9)
\uput[0](0.8,0.5){$\nev{T}$}
\uput[0](0.5,1.7){$\demidisque{T}$}
\uput[0](2,2.5){$\demiplan$}
\uput[0](3.5,0){$\re t$}
\uput[0](-0.1,2.8){$\im t$}
\uput[0](0.5,-0.6){$\frac{T}{2}$}
\end{pspicture}
 \end{figure}

Let $\kappa > 0$. Let $\cylindre{\kappa} \assign \lbc{1} z \in \mathbbm{C}
\ve{0} d (z, [0, + \infty [) < \kappa \rbc{1}$ \ and \ \ $\parabole{\kappa}
\assign \lbc{1} z \in \mathbbm{C} | | \mathcal{I} mz^{1 / 2} |^2 < \kappa
\rbc{1} = \lbc{1} z \in \mathbbm{C} | \mathcal{R} ez > \frac{1}{4 \kappa}
\mathcal{I} m^2 z - \kappa \rbc{1}$. $\parabole{\kappa}$ is the interior of a
parabola which contains $\cylindre{\kappa}$.
\FloatBarrier

\begin{figure}[h]
\caption{\label{parabole}}
\centering
  \begin{pspicture}(-2,-2.7)(6,3)
\psline[linewidth=0.1mm]{->}(-2.5,0)(5,0)
\psline[linewidth=0.1mm]{->}(0,-2)(0,2.3)
\parametricplot{0}{-180}{0.9 t sin mul 0.9 t cos mul}
\psline(0,0.9)(4,0.9)
\psline(0,-0.9)(4,-0.9)
\parametricplot{-2.5}{2.5}{-0.9 0.65 t t mul mul add t}
\psline[linewidth=0.1mm]{<->}(0,0)(-0.4,0.8)
\uput[0](5,0){$\re \tau$}
\uput[0](-0.3,2.5){$\im \tau$}
\uput[0](1,0.6){$\cylindre{\kappa}$}
\uput[0](2,1.6){$\parabole{\kappa}$}
\uput[0](-0.7,0.35){$\kappa$}
\end{pspicture}
\end{figure}

Let $\Omega$ be an open domain in $\mathbbm{C}^m$ and let $F$ be a complex
finite dimensional space. We denote by $\mathcal{A} (\Omega)$ the space of
$F$-valued analytic functions on $\Omega$, if there is no ambiguity on $F$.

Let $\mathfrak{B}$ denote the collection of all Borel sets on
$\mathbbm{R}^m$. An $F$-valued measure $\mu$ on $\mathbbm{R}^m$ is an
$F$-valued measurable function on $\mathfrak{B}$ satisfying the classical
countable additivity property (cf. [Ru]). Let $| \cdot |$ be a norm on $F$. We
denote by $| \mu |$ the positive measure defined \ by
\[ | \mu | (E) = \sup \sum_{j = 1}^{\infty} | \mu (E_j) | (E \in
   \mathfrak{B}), \]
the supremum being taken over all partition $\{E_j \}$ of $E$. In particular,
$| \mu | (\mathbbm{R}^m) < \infty$. Let us remark that $d \mu = h d| \mu |$
where $h$ is some $F$-valued function satisfying $|h| = 1$ $| \mu |$-a.e. If
$f$ is an $F$-valued measurable function on $\mathbbm{R}^m$ and $\lambda$ is a
positive measure on $\mathbbm{R}^m$ such that $\int_{\mathbbm{R}^m} |f| d
\lambda < \infty$, one can define an $F$-valued measure $\mu$ by setting $d
\mu = f d \lambda$. Then $d| \mu | = |f |d \lambda$.

We work with finite dimensional spaces of square matrices. We always consider
multiplicative norms on these spaces ($|AB| \leqslant |A\|B|$, for $A$ and $B$
square matrices) and we assume that $| \mathbbm{1} | = 1$. For $A = (a_{i,
j})_{1 \leqslant i, j \leqslant d}$ with $a_{i, j} \in \mathbbm{C}$, we set
$A^{\ast} = ( \bar{a}_{j, i})_{1 \leqslant i, j \leqslant d}$. For $\lambda,
\mu \in \mathbbm{C}^{\nu}$, we denote $\lambda \cdot \mu \assign \lambda_1
\mu_1 + \cdots + \lambda_{\nu} \mu_{\nu}$, $\bar{\lambda} \assign (
\bar{\lambda}_1, \ldots, \bar{\lambda}_{\nu})$, $\mathcal{I} m \lambda \assign
( \mathcal{I} m \lambda_1, \ldots, \mathcal{I} m \lambda_{\nu})$, $\lambda^2
\assign \lambda \cdot \lambda$, $| \lambda | \assign (\lambda \cdot
\bar{\lambda})^{1 / 2}$ (if $\lambda \in \mathbbm{R}^{\nu}$, $| \lambda | =
\sqrt{\lambda^2}$). In the whole paper, sums indexed by an empty set are, by
convention, equal to zero.

Here is an improved version of a theorem of Watson. This is in fact a theorem
of Nevanlinna, rediscovered by Sokal [So]. It gives a concise presentation of
what we need about Borel summation. In what follows, functions are defined on
some subset of $\mathbbm{C}$ and take their values in a complex finite
dimensional space $F$. First, we need:

\begin{definition}
  Let $\kappa > 0$ and $T > 0$. We say that
  \begin{itemizedot}
    \item A function $f$ satisfies $\mathcal{P}_{\kappa, T}$ if and only if
    $f$ is analytic on {\nev{T}} (fig. \ref{cercle}) and there exist $a_0,
    a_1, \ldots \in F$, $R_0, R_1, \ldots$ analytic functions on $\nev{T}$
    such that, for every $r \geqslant 0${\footnote{If $r = 0$, the expansion
    must be read $f (t) = R_0 (t)$, by the previous convention.}} and $t \in
    \nev{T}$,
    \[ \text{$f (t) = a_0 + \cdots + a_{r - 1} t^{r - 1} + R_r (t)$}, \]
    and for every $\bar{\kappa} < \kappa$, $\bar{T} < T$, there exists $K > 0$
    such that, for every $r \geqslant 0$ and $t \in \nev{\bar{T}}$,
    \[ |R_r (t) | \leqslant K \frac{r!}{\bar{\kappa}^r} |t|^r \]
    \item A function{\footnote{In general, we denote functions defined on the
    Borel plane by a hat.}} $\text{$\hat{f}$}$ satisfies $\mathcal{Q}_{\kappa,
    T}$ if and only if $\hat{f}$ is analytic on $\cylindre{_{\kappa}}$ (fig.
    \ref{parabole}){\tmstrong{}} and for every $\bar{\kappa} < \kappa$,
    $\bar{T} < T$, there exists $K > 0$ such that, for every $\tau \in
    \cylindre{_{\bar{\kappa}}}$
    \begin{equation}
      \label{barbara2} | \hat{f} (\tau) | \leqslant Ke^{^{} \frac{| \tau
      |}{\bar{T}}} .
    \end{equation}
  \end{itemizedot}
\end{definition}

This definition is justified by the following theorem and remark. Note that,
if $f$ satisfies $\mathcal{P}_{\kappa, T}$, the coefficients $a_n$ are
uniquely determined.

\begin{theorem}
  \label{theorembarbara0} Let $\kappa > 0$ and $T > 0$.
  \begin{itemizedot}
    \item If $f$ verifies $\mathcal{P}_{\kappa, T}$, then
    \begin{equation}
      \label{barbara4} \text{$\hat{f} (\tau) \assign \sum_{r = 0}^{\infty}
      \frac{a_r}{r!} \tau^r$}
    \end{equation}
    admits an analytic continuation on $\cylindre{\kappa}$ which verifies
    $\mathcal{Q}_{\kappa, T}$.
    
    \item If $\hat{f}$ verifies $\mathcal{Q}_{\kappa, T}$, then
    \begin{equation}
      \label{barbara6} \text{$f (t) \assign \int_0^{+ \infty} \hat{f} (\tau)
      e^{- \frac{\tau}{t}} \frac{d \tau}{t}$}
    \end{equation}
    verifies $\mathcal{P}_{\kappa, T}$.
    
    \item $\hat{f}$ given by (\ref{barbara4}) is called the Borel transform of
    $f$. $f$ given by (\ref{barbara6}) is called the Laplace transform of
    $\hat{f}$. These two transforms are inverse each to other.
  \end{itemizedot}
\end{theorem}

\begin{remark}
  \label{remarkbarbara2}If $\mathcal{P}_{\kappa, T}$ holds, the knowledge of
  the coefficients $a_0, a_1, \ldots$ allows one to recover $f$: $f$ is the
  Laplace transform of its Borel transform which only depends on $a_0, a_1,
  \ldots$ by (\ref{barbara4}). The shape of the domain $\nev{T}$ is crucial.
  For instance, the conclusion of this remark may fail if $\nev{T}$ is
  replaced by a truncated cone like $C_{T, \alpha} \assign \{t = re^{i \theta}
  \in \mathbbm{C} || \theta | < \alpha, r < T\}$ with $T > 0$, $\alpha <
  \frac{\pi}{2}$.
  
  Under our assumptions on $V$, the conjugate heat kernel will verify at least
  $\mathcal{P}_{\kappa, T}$ for some $\kappa > 0$ and $T > 0$. We do not prove
  a resummation estimate on a truncated cone $C_{T, \alpha}$ with $T > 0$,
  $\alpha > \frac{\pi}{2}$ and in general, the conjugate heat kernel will not
  verify the assumptions of Watson's theorem (cf. [So]).
\end{remark}

\begin{definition}
  \label{definitionbarbara2.5}Let $\tilde{a}_1, \ldots, \tilde{a}_r, \ldots
  \in F$. One says that the formal power series $\tilde{f} = \sum_{r \geqslant
  0} \tilde{a}_r t^r$ is Borel summable if there exist $\kappa, T > 0$ and a
  function $f$ satisfying $\mathcal{P}_{\kappa, T}$ such that, for every $r
  \geqslant 0$, $a_r = \tilde{a}_r$. f is called the Borel sum of $\tilde{f}$.
\end{definition}

\section{Main results}

The following theorem concernes the free case. In this case, it is possible to
give precise properties of the Borel transform of the conjugate heat kernel.
In particular, this Borel transform is analytic on the complex plane and is
exponentially dominated by the square root of the Borel variable on parabolic
domains which are symmetric with respect to the positive real axis.

\begin{theorem}
  \label{theorembarbara1}Let $\varepsilon > 0$. Let $\mu$ be a measure on
  $\mathbbm{R}^{\nu}$ with values in a complex finite dimensional space of
  square matrices verifying
  \[ \int_{\mathbbm{R}^{\nu}} \exp (\varepsilon \xi^2) d| \mu | (\xi) < \infty
     . \]
  Let $c (x) = \int \exp (ix \cdot \xi) d \mu (\xi)$ and let $u$ be the
  solution of
  \begin{equation}
    \label{barbara10} \left\{ \begin{array}{l}
      \partial_t u \text{$= \partial^2_x u$} + c (x) u\\
      \\
      u_{} |_{t = 0^+} =_{} \delta_{x = y} \mathbbm{1}
    \end{array} \right.
  \end{equation}
  Let $v$ be defined by u=$(4 \pi t)^{- \nu / 2} e^{- \frac{(x - y)^2}{4 t}}
  v$. Then $v$ admits a Borel transform $\hat{v}$ (with respect to $t$) which
  is analytic on $\mathbbm{C}^{1 + 2 \nu}$. Let $\kappa, R > 0$ and let
  \[ C \assign 2 \lp{2} \int_{\mathbbm{R}^{\nu}} \exp \lp{1} \frac{2
     \kappa}{\varepsilon} + \frac{\varepsilon}{2} \xi^2 + R| \xi | \rp{1} d|
     \mu | (\xi) \srp{2}{1 / 2}{} . \]
  Then, for every $(\tau, x, y) \in \parabole{\kappa} \times \mathbbm{C}^{2
  \nu}$ such that $| \mathcal{I} mx| < R$ and $| \mathcal{I} my| < R$,
  \begin{equation}
    \label{barbara12} | \hat{v} (\tau, x, y) | \leqslant \exp \lp{1} C| \tau
    |^{1 / 2} \rp{1} .
  \end{equation}
\end{theorem}

\begin{remark}
  \label{remarkbarbara4}$\parabole{\kappa}$ is the interior of a parabola
  which contains {\cylindre{{\kappa}}} (fig. \ref{parabole}). Let $T > 0$.
  Estimate (\ref{barbara12}) is better than (\ref{barbara2}). Then $v (t, x,
  y)$ verifies $\mathcal{P}_{\kappa, T}$: the small time expansion of the
  conjugate heat kernel is Borel summable and its Borel sum is equal to $v$.
\end{remark}

The following corollary deals with the partition function on the torus and is
a consequence of the Theorem \ref{theorembarbara1}.

\begin{corollary}
  \label{corollairebarbara3}Let $\varepsilon > 0$ and $d \in
  \mathbbm{N}^{\ast}$. For each $q \in \mathbbm{Z}^{\nu}$, let $c_q$ be a
  square matrix acting on $\mathbbm{C}^d$. Assume that $c_{- q} = c^{\ast}_q$
  and
  \[ \sum_{q \in \mathbbm{Z}^{\nu}} e^{4 \pi^2 \varepsilon q^2} |c_q | <
     \infty . \]
  Let $c (x) \assign \sum_{q \in \mathbbm{Z}^{\nu}} c_q e^{2 i \pi q.x}$. Let
  $\lambda_1 \leqslant \lambda_2 \leqslant \cdots$ be the eigenvalues of the
  operator $H \assign - \partial^2_x - c (x)$ acting on $\mathbbm{C}^d$-valued
  functions defined on the torus $( \mathbbm{R} / \mathbbm{Z})^{\nu}$. For
  each $q \in \mathbbm{Z}^{\nu}$, there is a function $\hat{w} (q, .)$
  analytic on $\mathbbm{C}$ satisfying
  \begin{itemize}
    \item For every $\kappa > 0$, there exist constants $C_1 > 0$ and $C_2 >
    0$ such that, for every $\text{$q \in \mathbbm{Z}^{\nu}$ and $\tau \in
    \parabole{\kappa}$},$
    \begin{equation}
      \label{barbara16b} \text{$| \hat{w} (q, \tau) | \leqslant C_1 \exp
      \lp{1} C_2 | \tau |^{1 / 2} \rp{1}$} .
    \end{equation}
    \item For every $t \in \demiplan{}$
    \begin{equation}
      \label{barbara16a} \sum_{n = 1}^{+ \infty} e^{- \lambda_n t} = (4 \pi
      t)^{- \nu / 2} \sum_{q \in \mathbbm{Z}^{\nu}} e^{- \frac{q^2}{4 t}}
      \int_0^{+ \infty} e^{- \frac{\tau}{t}} \hat{w} (q, \tau) \frac{d
      \tau}{t} .
    \end{equation}
  \end{itemize}
\end{corollary}

\begin{remark}
  By the argument of Remark \ref{remarkbarbara4}, Corollary
  \ref{corollairebarbara3} implies the following result. Let $c$ as in
  Corollary \ref{corollairebarbara3}. For each $q \in \mathbbm{Z}^{\nu}$,
  there are numbers $a_{1, q}, a_{2, q}, \ldots \in \mathbbm{C}$ and functions
  $R_{0, q}, R_{1, q}, \ldots \in \text{$\mathcal{A} ( \demiplan{})$}$ such
  that, for every $r \geqslant 0$ and $t \in \demiplan{}$,
  \begin{equation}
    \label{barbara18} \sum_{n = 1}^{+ \infty} e^{- \lambda_n t} = (4 \pi t)^{-
    \nu / 2} \sum_{q \in \mathbbm{Z}^{\nu}} e^{- \frac{q^2}{4 t}} \lp{1} d +
    a_{1, q} t + \cdots + a_{r - 1, q} t^{r - 1} + R_{r, q} (t) \rp{1},
  \end{equation}
  and for each $T, \kappa > 0$, there exist $K > 0$ such that
  \[ |R_{r, q} (t) | \leqslant K \frac{r!}{\kappa^r} |t|^r, \]
  for every $r \geqslant 0, q \in \mathbbm{Z}^{\nu}, t \in \nev{T} .$
  
  The index $q$ in (\ref{barbara18}) can be viewed as labeling closed
  classical trajectories (geodesics) on the torus. Then $q^2$ denotes the
  length of such a geodesic. All coefficients of the expansion \ref{barbara18}
  have classical (or geometric) interpretation.
\end{remark}

The following theorem deals with the harmonic case (i.e. $V (x) = \alpha x^2 -
c (x)$ with $\alpha \in \mathbbm{R}$) and gives a statement about the
expansion of the conjugate heat kernel which provides Borel summability. But
we do not obtain a statement as precise as in Theorem \ref{theorembarbara1}
about its Borel transform (the proof is established without working in the
Borel plane).

\begin{theorem}
  \label{theorembarbara3}Let $\omega \in \mathbbm{R} \cup i \mathbbm{R}$. Let
  $c$ be as in Theorem \ref{theorembarbara1}. Let u be the solution of
  \begin{equation}
    \label{barbara20} \left\{ \begin{array}{l}
      \partial_t \text{$u = \lp{1} \partial^2_x - \frac{\omega^2}{4} x^2
      \rp{1} u + c (x) u$}\\
      \\
      u_{} |_{t = 0^+} =_{} \delta_{x = y} \mathbbm{1}
    \end{array} \right.
  \end{equation}
  
  \begin{itemizedot}
    \item Then there are a number $T > 0$, functions $a_1, a_2, \ldots \in
    \mathcal{A} ( \mathbbm{C}^{2 \nu})$ and $R_0, R_1, \ldots \in
    \text{$\mathcal{A} ( \demidisque{T} \times \mathbbm{C}^{2 \nu})$}$ such
    that
    \begin{equation}
      \label{barbara21} u = (4 \pi t)^{- \nu / 2} e^{- \frac{(x - y)^2}{4 t}}
      \lp{1} \mathbbm{1} + a_1 (x, y) t + \cdots + a_{r - 1} (x, y) t^{r - 1}
      + R_r (t, x, y) \rp{1},
    \end{equation}
    for every $r \geqslant 0$, $t \in \demidisque{T}$ and $(x, y) \in
    \mathbbm{C}^{2 \nu}$.
    
    \item And for each $R > 0$, there exist $K > 0$ and $\kappa > 0$ such
    that,
    \begin{equation}
      \text{\label{barbara22}} |R_r (t, x, y) | \leqslant K
      \frac{r!}{\kappa^r} |t|^r,
    \end{equation}
    for every $r \geqslant 0$, $t \in \demidisque{T}$, $(x, y) \in
    \mathbbm{C}^{2 \nu}$, $|x| < R$, $|y| < R$.
  \end{itemizedot}
\end{theorem}

\begin{remark}
  {\tmdummy}
  
  \begin{itemize}
    \item If $\omega \in \mathbbm{R}$, the heat kernel is uniquely defined as
    the kernel of an analytic semi-group. In the case $\omega \in
    i\mathbbm{R}$, one can also give a uniqueness statement (see [Ha7]) which
    allows one to speak about ``the'' heat kernel.
    
    \item By (\ref{barbara22}) and Theorem \ref{theorembarbara0}, the
    expansion in (\ref{barbara21}) admits a Borel transform satisfying
    $\mathcal{Q}_{\kappa, T}$. In the case $\omega = 0$, the estimate
    (\ref{barbara22}) is better than that obtained by Remark
    \ref{remarkbarbara4} since {\nev{T}}$\subset \demidisque{T}$ (fig.
    \ref{cercle}). In particular we get a uniform estimate when $\mathcal{R}
    et \rightarrow 0^+$ and $\im t$ is a non-vanishing constant. Hence, we
    obtain information about the Schr\"odinger kernel. In fact, this
    information is contained in Theorem \ref{theorembarbara1}: it is no
    difficult to see that Theorem \ref{theorembarbara1} implies the estimate
    (\ref{barbara22}) for $t \in \text{$\demidisque{T}$} \mathbbm{}$.
  \end{itemize}
\end{remark}

\section{\label{proof_theorem}Proof of the theorems}

The proof of our result uses an explicit formula of the heat kernel which has
two expressions (compare (\ref{maria3}) and (\ref{anna3})). This formula is
based on a matrix ($\Omega^{\natural}$ or $\Omega$) and a path
($q^{\natural}_{\omega}$ or $q_{\omega}$). There are two cases. In the free
case ($\omega = 0$), the matrix $\Omega$ is linear in $t$ and the path
$q_{\omega}$ does not depend on $t$. Consequently, the proof of our results is
simple and working in the Borel plane is natural. In the harmonic case
($\omega \neq 0$), we must deal with the $t$-dependence of the matrix and the
path. The following subsection is devoted to the study of this matrix (let us
emphasize that only Lemma \ref{lemmaclaudia1}, in this subsection, is useful
for the study of the free case).

\subsection{The deformation matrix}

Let $n \in \mathbbm{N}$, $\omega \in \mathbbm{C}$, $t \in] 0, + \infty [$ such
that $| \omega t| < \pi$ and let $\text{$(s_1, \ldots, s_n) \in [0, t]$}$ such
that $0 < s_1 < \cdots < s_n < t$. For $A \in \mathbbm{C}$, denote $\sh A
\assign \frac{1}{2} (e^A - e^{- A})$. In the following sections, the matrix
\begin{equation}
  \label{claudia0} \Omega^{\natural} \assign \text{$\left( \frac{\tmop{sh}
  (\omega s_{j \wedge k}) \tmop{sh} \lp{1} \omega (t - s_{j \vee k})
  \rp{1}}{\omega \tmop{sh} (\omega t)} \right)_{1 \leqslant j, k \leqslant
  n}$}
\end{equation}
plays an important role (cf. (\ref{anna6}) and (\ref{maria3})) and some of its
properties must be established. However, since we shall also consider complex
values of $t$, we study the following matrix. Assume now that $t \in
\mathbbm{C}$ . Let $\text{$(s_1, \ldots, s_n) \in [0, 1]$}$ such that $0 < s_1
< \cdots < s_n < 1$. Let $\Omega$ be defined by
\begin{equation}
  \label{claudia2} \Omega^{} : = \left( \frac{\tmop{sh} (\omega ts_{j \wedge
  k}) \tmop{sh} \lp{1} \omega t (1 - s_{j \vee k}) \rp{1}}{\omega \tmop{sh}
  (\omega t)} \right)_{1 \leqslant j, k \leqslant n} .
\end{equation}
The goal of this section is to study some properties of $\Omega$ (Proposition
\ref{propositionclaudia4}) in particular when $n$, the dimension of the
matrix, is large. This control, in large dimension, plays an important role in
the proof of the Borel summability of the heat kernel expansion. Note that the
connection of this matrix with a propagator (cf. (\ref{maria10})) seems to be
important for the understanding of its properties. The following lemma gives a
useful elementary property of $\bar{\Omega} \assign \frac{1}{t} \Omega
|_{\omega = 0} = \lp{1} s_{j \wedge k} (1 - s_{j \vee k}) \rp{1}_{1 \leqslant
j, k \leqslant n}$ and Proposition \ref{propositionclaudia4}, the goal of this
subsection, can be viewed as a generalization of this lemma. For $(\xi_1,
\ldots, \xi_n) \in \mathbbm{R}^{\nu n}$, let
\begin{equation}
  \label{claudia4} \bar{\Omega} \cdot \xi \otimes_n \xi : = \sum_{j, k = 1}^n
  s_{j \wedge k} (1 - s_{j \vee k}) \xi_j \cdot \xi_k .
\end{equation}
\begin{lemma}
  \label{lemmaclaudia1}For every $n \geqslant 1, (\xi_1, \ldots, \xi_n) \in
  \mathbbm{R}^{\nu n}$ and $(s_1, \ldots, s_n) \in [0, 1]$ such that $0 < s_1
  < \cdots < s_n < 1$
  \[ 0 \leqslant \bar{\Omega} \cdot \xi \otimes_n \xi \leqslant n \sum_{j =
     1}^n \xi_j^2 . \]
\end{lemma}

\begin{proof}
  The upper bound of the quantity $\bar{\Omega} \cdot \xi \otimes_n \xi$ is
  elementary and its positivity can be viewed as a consequence of Remark
  \ref{remarkclaudia2}.
\end{proof}

\begin{lemma}
  \label{lemmaclaudia2}Let $\tilde{T} > 0$ and $M > 0$. There exists $T > 0$
  satisfying the following property. Let $f$ be an analytic function $f$ on
  $\disque{\tilde{T}}$ verifying $f (0) = 0$, $f' (0) = 1$, $\sup_{t \in
  \text{$\disque{\tilde{T}}$}} \text{$|f (t) | \leqslant M$}$ and, for every
  $t \in \disque{\tilde{T}}$,
  \begin{equation}
    \label{claudia7.5} \mathcal{R} et = 0 \Rightarrow \mathcal{R} ef (t) = 0.
  \end{equation}
  Then, for $t \in \text{$\disque{T}$}$,
  \begin{equation}
    \label{claudia8} \mathcal{R} et > 0 \Rightarrow \mathcal{R} ef (t) > 0.
  \end{equation}
\end{lemma}

\begin{proof}
  We can choose $T > 0$, depending \tmtextbf{only} on $\tilde{T}$ and $M$,
  such that every analytic function $f$ verifying $\text{$f (0) = 0$}$, $f'
  (0) = 1$ and $|f (t) | \leqslant M$ for $t \in \text{$\disque{\tilde{T}}$}$
  is a one-to-one analytic mapping on $\disque{T}$. For small $t > 0$,
  $\mathcal{R} ef (t) > 0$ since $\text{$f (0) = 0$}$ and $f' (0) = 1$. Then
  by (\ref{claudia7.5}) and since $f$ is a one-to-one analytic mapping, one
  gets (\ref{claudia8}).
\end{proof}

\begin{proposition}
  \label{propositionclaudia2}Let $\mathcal{E}$ be the space of Borel measures
  $\mu = \sum_{j = 1}^n \delta_{s_j} \xi_j$ such that $n \geqslant 1, \xi_j
  \in \mathbbm{R^{\nu}}, s_j \in] 0, 1 [$. Let $z \in \disque{\pi}$. We denote
  by $(., .)_z$ \ the following bilinear form on $\mathcal{E}$
  \[ \text{$(\mu_1, \mu_2)_z$} \assign \int_0^1 \int_0^1 \frac{\tmop{sh} (zs
     \wedge s') \tmop{sh} \lp{1} z (1 - s \vee s') \rp{1}}{z \tmop{sh} z} d
     \mu_1 (s) \cdot d \mu_2 (s') . \]
  Note that
  \[ \text{$(\mu_1, \mu_2)_0$} = \int_0^1 \int_0^1 s \wedge s' (1 - s \vee s')
     d \mu_1 (s) \cdot d \mu_2 (s') . \]
  Then, for $z \in \disque{\pi}$,
  \begin{equation}
    \label{claudia10} \forall \mu \in \mathcal{E}, | (\mu, \mu)_z | \leqslant
    \frac{\pi^2}{\pi^2 - |z|^2} (\mu, \mu)_0 .
  \end{equation}
\end{proposition}

\begin{proof}
  For $(s, s') \in [0, 1]^2$, let
  \[ K (s, s') \assign \frac{\tmop{sh} (zs \wedge s') \tmop{sh} \lp{1} z (1 -
     s \vee s') \rp{1}}{z \tmop{sh} z} . \]
  Then
  \begin{equation}
    \label{claudia12} \left\{ \begin{array}{l}
      - \frac{d^2 K}{ds^2} + z^2 K = \delta_{s = s'}\\
      \\
      K|_{s = 0} = K|_{s = 1} = 0
    \end{array} \right. .
  \end{equation}
  Let $(\xi_1, \ldots, \xi_n) \in \mathbbm{R}^{\nu n}$ and $(s_1, \ldots, s_n)
  \in] 0, 1 [^n$. Let $u$ be the function on $[0, 1]$ defined by
  \[ u (s) = \sum_{j = 1}^n \frac{\tmop{sh} (zs \wedge s_j) \tmop{sh} \lp{1} z
     (1 - s \vee s_j) \rp{1}}{z \tmop{sh} z} \xi_j . \]
  $u$ is continuous and piecewise differentiable on $[0, 1]$. Let $\mu \assign
  \sum_{j = 1}^n \delta_{s_j} \xi_j$. By (\ref{claudia12})
  \begin{equation}
    \label{claudia16} \left\{ \begin{array}{l}
      - \frac{d^2 u}{ds^2} + z^2 u = \mu\\
      \\
      u (0) = u (1) = 0
    \end{array} \right. .
  \end{equation}
  For $n \geqslant 1$, $k, l \in \{1, \ldots, \nu\}$ and $s \in [0, 1]$, let
  $e_{n, k, l} (s) \assign \sqrt{2} \sin (n \pi s) \delta_{k = l}$. $(e_{n,
  k})_{n, k}$ is an orthonormal basis of $L^2 ([0, 1], \mathbbm{C}^{\nu})$
  which diagonalizes the unbounded operator $S \assign - \frac{d^2}{d s^2}$.
  Let
  \begin{eqnarray*}
    H^1_0 \assign &  & \lbc{2} f \in L^2 ([0, 1], \mathbbm{C}^{\nu}) \ve{2}
    \frac{d f}{d s} \in L^2, f (0) = f (1) = 0 \rbc{2}\\
    = &  & \text{$\lbc{2} \sum_{n, k} f_{n, k} e_{n, k}$} \ve{2} \sum_{n, k}
    |n f_{n, k} |^2 < \infty \rbc{2},
  \end{eqnarray*}
  \[ H^{- 1} \assign \lbc{3} \sum_{n, k} f_{n, k} e_{n, k} \ve{3} \sum_{n, k}
     \ve{2} \frac{f_{n, k}}{n} \sve{2}{2}{} < \infty \rbc{3} . \]
  For $(f, g) \in H^{- 1} \times H_0^1$, let
  \[ \left\langle f, g \right\rangle \assign \int_0^1 \bar{f} (s) \cdot g (s)
     d s = \sum_{n, k} \bar{f}_{n, k} g_{n, k} . \]
  Note that $S$ can be viewed as an isomorphism from $H_0^1$ to $H^{- 1}$. The
  function $u$ belongs to $H_0^1$ and by (\ref{claudia16}), $\mu = (S + z^2) u
  \in H^{- 1}$. Moreover, $(\mu, \mu)_z$=$\int_0^1 u \cdot d \mu$ and $\pi^2$
  is the lowest eigenvalue of $S$ (every eigenvalue of $S$ are real). Hence,
  for $z \in \disque{\pi}$,
  \begin{equation}
    \label{claudia17} \text{$(\mu, \mu)_z = \left\langle \mu, (S + z^2)^{- 1}
    \mu \right\rangle,$} (\mu, \mu)_0 = \left\langle \mu, S^{- 1} \mu
    \right\rangle
  \end{equation}
  For $\lambda \in [\pi^2, + \infty [$ and $z \in \disque{\pi}$, the following
  inequality holds
  \begin{equation}
    \label{claudia17.5} | (\lambda + z^2)^{- 1} | \leqslant \frac{\pi^2}{\pi^2
    - |z|^2} \times \lambda^{- 1}
  \end{equation}
  Then (\ref{claudia10}) is a consequence of (\ref{claudia17}) and
  (\ref{claudia17.5}).
\end{proof}

Let $v (s) \assign \sum_{j = 1}^n s \wedge s_j (1 - s \vee s_j) \xi_j$. Then
\[ (\mu, \mu)_0 = \int_0^1 v \cdot d \mu = \int_0^1 \left( \frac{d v}{d s}
   \right)^2 \]
since $\mu = - \frac{d^2 v}{ds^2}$. One gets

\begin{remark}
  \label{remarkclaudia2}$(., .)_0$ is a positive definite symmetric bilinear
  form on $\mathcal{E}$.
\end{remark}

For $(\xi_1, \ldots, \xi_n) \in \mathbbm{R}^{\nu n}$ and $0 < s_1 < \cdots <
s_n < 1$, let
\begin{equation}
  \label{claudia29} \Omega . \xi \otimes_n \xi \assign \sum_{j, k = 1}^n
  \frac{\tmop{sh} (\omega ts_{j \wedge k}) \tmop{sh} \lp{1} \omega t (1 - s_{j
  \vee k}) \rp{1}}{\omega \tmop{sh} (\omega t)} \xi_j \cdot \xi_k .
\end{equation}
\begin{proposition}
  \label{propositionclaudia4}Let $\omega \in \mathbbm{R} \cup i \mathbbm{R}$.
  There exists $T_d > 0$ such that for every $\text{$n \geqslant 1, (s_1,
  \ldots, s_n) \in [0, 1]^n, (\xi_1, \ldots, \xi_n) \in \mathbbm{R}^{\nu n}$},
  t \in \mathbbm{C}$ with the condition $0 < s_1 < \cdots < s_n < 1, |t| <
  T_d$,
  \begin{equation}
    \label{claudia30} \mathcal{R} et \geqslant 0 \Rightarrow \mathcal{R} e
    \left( \Omega . \xi \otimes_n \xi \right) \geqslant 0,
  \end{equation}
  \begin{equation}
    \label{claudia32} \left| \Omega . \xi \otimes_n \xi \right| \leqslant 2
    n|t| \sum_{j = 1}^n \xi_j^2 .
  \end{equation}
\end{proposition}

\begin{proof}
  Let $\mu \assign \sum_{j = 1}^n \delta_{s_j} \xi_j$. Then $\Omega . \xi
  \otimes_n \xi = t (\mu, \mu)_{\omega t}$ and $\bar{\Omega} \cdot \xi
  \otimes_n \xi = (\mu, \mu)_0$. Let $\tilde{T} \assign \frac{\pi}{\sqrt{2} |
  \omega |}$. By (\ref{claudia10}) and Lemma \ref{lemmaclaudia1},
  (\ref{claudia32}) holds for $t \in \disque{\tilde{T}}$. Note that
  (\ref{claudia32}) can also be obtained directly without using
  (\ref{claudia10}).
  
  Let us choose arbitrary $\xi_1, \ldots, \xi_n$ such that $(\xi_1, \ldots,
  \xi_n)$ does not vanish. By Remark \ref{remarkclaudia2}, $(\mu, \mu)_0 \neq
  0$. Let $f$ be the function defined by
  \[ f (t) = t \frac{(\mu, \mu)_{\omega t}}{(\mu, \mu)_0} = \frac{\Omega . \xi
     \otimes_n \xi}{(\mu, \mu)_0} . \]
  By (\ref{claudia10}), $f$ is bounded on $\text{$\disque{\tilde{T}}$}$ by $2
  \times \frac{\pi}{\sqrt{2} | \omega |}$. Since $f$ is odd and $\omega \in
  \mathbbm{R} \cup i \mathbbm{R}$, $\overline{f (t)} = f ( \bar{t})$, for $t
  \in \text{$\text{$\disque{\tilde{T}}$}$}$. Then $f (t) = t + a t^3 + b t^5 +
  \cdots$ with $a, b, \ldots$ real and (\ref{claudia7.5}) holds. Obviously, $f
  (0) = 0$ and $f' (0) = 1$. Then we can use Lemma \ref{lemmaclaudia2} and
  there exists $T_d < \tilde{T}$ such that $\re f (t) > 0$ for $t \in
  \demidisque{T_d}$. Since $(\mu, \mu)_0 \in] 0, + \infty [$,
  (\ref{claudia30}) holds for $t \in \demidisque{T_d}$.
\end{proof}

\subsection{The deformation formula}

Let $(x, y) \in \mathbbm{C}^{2 \nu}$, $t \in \mathbbm{C}^{\ast}$ and $\omega
\in \mathbbm{R} \cup i\mathbbm{R}$. Let $u_{\omega}$ be defined by
\begin{equation}
  \label{anna0.0} u_{\omega} = \left( 4 \pi \frac{\tmop{sh} (\omega
  t)}{\omega} \right)^{- \nu / 2} \exp \lp{2} - \frac{1}{4}
  \frac{\omega}{\tmop{sh} (\omega t)} (\tmop{ch} (\omega t) (x^2 + y^2) - 2 x
  \cdot y) \rp{2} .
\end{equation}

Then, by one variant of Mehler's formula,
\[ \text{$\text{$\partial_t u_{\omega} = \lp{1} \partial^2_x -
   \frac{\omega^2}{4} x^2 \rp{1} u_{\omega}$}, u_{\omega} |_{t = 0^+} =_{}
   \delta_{x = y} .$} \]
Note that
\[ u_0 = \left( 4 \pi t \right)^{- \nu / 2} e^{- \frac{(x - y)^2}{4 t}} . \]
We denote by $q_{\omega}^{\natural}$ the extremal path of the action $S
\assign \int_0^t \lp{1} \dot{q}^2 (s) + \omega^2 q^2 (s) \rp{1} ds$ such that
$q_{\omega}^{\natural} (0) = y$ and $q_{\omega}^{\natural} (t) = x$. Then
\begin{equation}
  \label{anna03} q_{\omega}^{\natural} (s) = \frac{1}{\tmop{sh} (\omega t)}
  \lp{2} \tmop{sh} (\omega s) x + \tmop{sh} \lp{1} \omega (t - s) \rp{1} y
  \rp{2} .
\end{equation}
Let $(\xi_1, \ldots, \xi_n) \in \mathbbm{R}^{\nu n}$. Let us define, for $s
\in [0, 1]$,
\[ q_{\omega} (s) \assign \frac{1}{\tmop{sh} (\omega t)} \lp{2} \tmop{sh}
   (\omega t s) x + \tmop{sh} \lp{1} \omega t (1 - s) \rp{1} y \rp{2} \]
and, for $s = (s_1, \ldots, s_n) \in [0, 1]^n$,
\[ q_{\omega}^n (s) \cdot \xi \assign q_{\omega} (s_1) \cdot \xi_1 + \cdots +
   q_{\omega} (s_n) \cdot \xi_n . \]
\begin{remark}
  \label{anna04}There exists $\delta > 0$ such that, for every $(t, \omega)
  \in \mathbbm{C}^2$, $s \in [0, 1]$, ($x, y) \in \mathbbm{C}^{2 \nu}$,
  \begin{equation}
    \label{anna04.a} |q_{\omega} (s) | \leqslant 4 \max (|x|, |y|) \text{if }
    | \omega t| < \delta .
  \end{equation}
\end{remark}

\begin{proposition}
  \label{annaprop1}Let $\mu$ be a measure on $\mathbbm{R}^{\nu}$ with values
  in a complex finite dimensional space of square matrices. Let us assume that
  for every $R > 0$
  \begin{equation}
    \text{} \label{anna1} \int_{\mathbbm{R}^{\nu}} \exp (R| \xi |) d| \mu |
    (\xi) < \infty .
  \end{equation}
  Let $c (x) \assign \int_{\mathbbm{R}^{\nu}} \exp (ix \cdot \xi) d \mu
  (\xi)$. Let $\omega \in \mathbbm{R} \cap i \mathbbm{R}$. Let $v$ be defined
  by
  \begin{equation}
    \label{anna3} \left\{ \begin{array}{l}
      \text{$v = \mathbbm{1} + \sum_{n \geqslant 1} v_n$},\\
      \\
      v_n (t, x, y) = t^n \int_{0 < s_1 < \cdots < s_n < 1}
      \int_{\mathbbm{R}^{\nu n}} e^{- \Omega . \xi \otimes_n \xi}
      e^{iq_{\omega}^n (s) \cdot \xi} d^{\nu n} \mu^{\otimes} (\xi) d^n s.
    \end{array} \right.
  \end{equation}
  In (\ref{anna3}), $d^n s$ denotes $ds_1 \cdots d \nonesep s_n$ and $d^{\nu
  n} \mu^{\otimes} (\xi)$ denotes $d \mu (\xi_n) \cdots d \mu (\xi_1)$.
  
  Then there exists $T_c > 0$ such that $v \in \mathcal{A} ( \demidisque{T_c}
  \times \mathbbm{C}^{2 \nu})$ and the function $u : = u_{\omega} v$ is
  solution of (\ref{barbara20}). If $\omega = 0$, $v \in \mathcal{A} (
  \demiplan{} \times \mathbbm{C}^{2 \nu})$ .
\end{proposition}

\begin{proof}
  Let $T_d$ given by Proposition \ref{propositionclaudia4} and let $\delta$
  given by Remark \ref{anna04}. We choose
  \begin{equation}
    \label{anna4} \text{$T_c = \min \lp{2} T_d, \frac{\delta}{| \omega |}
    \rp{2}$} .
  \end{equation}
  Let $R > 0$ and let
  \[ A \assign \int \exp (4 R| \xi |) d| \mu | (\xi) . \]
  Let $(x, y) \in \mathbbm{C}^{2 \nu}$ such that $\text{$|x|, |y| < R$}$. Let
  $t \in \demidisque{T_c}$. By (\ref{claudia30}) and (\ref{anna04.a})
  \begin{equation}
    \label{anna4.5} | \exp (- \Omega . \xi \otimes_n \xi) \exp (iq_{\omega}^n
    (s) \cdot \xi) | \leqslant \exp \lp{1} 4 R (| \xi_1 | + \cdots + | \xi_n
    |) \rp{1} .
  \end{equation}
  Hence
  \[ |v_n | \leqslant \frac{(|t|A)^n}{n!}, \]
  so the series in (\ref{anna3}) converges absolutely. Since $R$ is arbitrary,
  $v$ is well defined on \ \ $\demidisque{T_c} \times \mathbbm{C}^{2 \nu}$. By
  dominated convergence theorem, one can check that $v_n$ and hence $v \in
  \text{$\mathcal{A} ( \demidisque{T_c} \times \mathbbm{C}^{2 \nu}$}$).
  
  Let us verify that the function $u : = u_{\omega} v$, with $v$ given by
  (\ref{anna3}), is solution of (\ref{barbara20}). A solution $u$ of
  (\ref{barbara20}) is given, if we use the relation $u = u_{\omega} v$ by a
  solution $v$ of the conjugate equation
  \[ \left\{ \begin{array}{l}
       \lp{2} \text{$\partial_t - \frac{2}{u_{\omega}} \partial_x u_{\omega}
       \cdot \partial_x \rp{2} v = \partial_x^2 v$} + c (x) v \text{ \ \ \ (}
       t \neq 0 \text{)}\\
       \\
       v_{} |_{t = 0^+} =_{} \mathbbm{1}
     \end{array} \right. . \]
  Notice that
  \[ - \frac{2}{u_{\omega}} \partial_x u_{\omega} = \frac{\omega}{\tmop{sh}
     (\omega t)} \lp{1} \tmop{ch} (\omega t) x - y \rp{1} =
     \dot{q}_{\omega}^{\natural} (t) . \]
  Set $v_0 \assign \mathbbm{1} .$ It is then sufficient to verify that $v_n$
  given by (\ref{anna3}) satisfies
  \begin{equation}
    \label{anna5} \left\{ \begin{array}{l}
      \lp{1} \partial_t + \dot{q}_{\omega}^{\natural} (t) \cdot \partial_x
      \rp{1} v_n \text{$= \partial^2_x v_n$} + c (x) v_{n - 1}\\
      \\
      v_n |_{t = 0^+} =_{} 0
    \end{array} \right.,
  \end{equation}
  for $(t, x, y) \in \demidisque{T_c} \times \mathbbm{C}^{2 \nu}$, $n
  \geqslant 1$.
  
  It suffices to check (\ref{anna5}) for $(t, x, y) \in] 0, T_c [\times
  \mathbbm{C}^{2 \nu}$. By (\ref{anna3}) and the definition of $c$
  \[ v_n = \int_{0 < s_1 < \cdots < s_n < t} F_n d^n s, \]
  where
  \begin{equation}
    \label{anna6} F_n \assign \lb{2} \exp \lp{1} \sum_{j, k = 1}^n \Omega_{j,
    k}^{\natural} \partial_{z_j} \cdot \partial_{z_k} \rp{1} c (z_n) \cdots c
    (z_1) \rb{2} \ve{2}_{\tmscript{\begin{array}{l}
      \text{$z_1 = q_{\omega}^{\natural} (s_1)$}\\
      \text{$\cdots$}\\
      \text{$z_n = q_{\omega}^{\natural} (s_n)$}
    \end{array}}} .
  \end{equation}
  Here $\Omega^{\natural}$ denotes the \ $n \times n$ matrix defined by
  (\ref{claudia0}). In the Appendix, we give an heuristic explanation of this
  result. Here is a rigorous verification. We have
  
  $\text{$(\partial_t + \dot{q}_{\omega}^{\natural} (t) \cdot \partial_x) v_n
  = (\tmop{boundary}) + (\tmop{interior})$}$ where
  \[ (\tmop{boundary}) \assign \int_{0 < s_1 < \cdots < s_{n - 1} < t} F_n
     |_{s_n = t} d^{n - 1} s, \]
  \[ (\tmop{interior}) \assign \int_{0 < s_1 < \cdots < s_n < t} \lp{1}
     \text{$\partial_t + \dot{q}_{\omega}^{\natural} (t) \cdot \partial_x
     \rp{1} F_n$} d^n s. \]
  In particular, $\Omega_{j, k}^{\natural} |_{s_n = t} = 0$ for $j = n$ or $k
  = n$. Therefore
  \[ \text{$(\tmop{boundary}) = c (x) v_{n - 1}$.} \]
  Now we want to prove that $(\tmop{interior}) = \partial^2_x v_n$.
  
  Using, for instance, explicit expressions of $\dot{q}_{\omega}^{\natural}
  (t)$ and $q_{\omega}^{\natural}$, we get, for $s \in] 0, t [$ and $m = 1,
  \ldots, \nu$,
  \begin{equation}
    \label{anna7} \text{$\lp{1} \partial_t + \dot{q}_{\omega}^{\natural} (t)
    \cdot \partial_x \rp{1}$} \text{$q_{\omega, m}^{\natural} (s) = 0$} .
  \end{equation}
  Then, if $\varphi (z_1, \ldots, z_n$) is an arbitrary differentiable
  function of $(z_1, \ldots, z_n) \in \mathbbm{C}^{\nu n}$,
  \[ \lp{1} \partial_t \text{$+ \dot{q}_{\omega}^{\natural} (t) \cdot
     \partial_x \rp{1}$} \text{$[\varphi (z_1, \ldots, z_n$)]}
     |_{\tmscript{\begin{array}{l}
       \text{$z_1 = q_{\omega}^{\natural} (s_1)$}\\
       \text{$\cdots$}\\
       \text{$z_n = q_{\omega}^{\natural} (s_n)$}
     \end{array}}} = 0. \]
  So, since $\partial_t \Omega_{j, k}^{\natural} = \frac{\tmop{sh} (\omega
  s_j)}{\tmop{sh} (\omega t)} \frac{\tmop{sh} (\omega s_k)}{\tmop{sh} (\omega
  t)}$, denoting $B_{j, k} \assign \frac{\tmop{sh} (\omega s_j)}{\tmop{sh}
  (\omega t)} \frac{\tmop{sh} (\omega s_k)}{\tmop{sh} (\omega t)}$,
  \[ (\tmop{interior}) = \int_{0 < s_1 < \cdots < s_n < t} G_n d^n s, \]
  where
  \[ G_n \assign \lb{3} \lp{2} \sum_{j, k = 1}^n B_{j, k} \partial_{z_j} \cdot
     \partial_{z_k} \rp{2} \exp \lp{2} \sum_{j, k = 1}^n \Omega_{j,
     k}^{\natural} \partial_{z_j} \cdot \partial_{z_k} \rp{2} c (z_n) \cdots c
     (z_1) \rb{3} \ve{3}_{\tmscript{\begin{array}{l}
       \text{$z_1 = q_{\omega}^{\natural} (s_1)$}\\
       \text{$\cdots$}\\
       \text{$z_n = q_{\omega}^{\natural} (s_n)$}
     \end{array}}} . \]
  For $(\alpha, \beta) \in \{1, \ldots, \nu^{} \}^2$, one has
  $\partial_{x_{\beta}} q_{\omega, \alpha}^{\natural} = \delta_{\alpha =
  \beta} \times \frac{\sh (\omega s)}{\sh (\omega t)}$. Since
  $q_{\omega}^{\natural}$ is linear with respect to $x$,
  \[ \partial^2_x [c (z_n) \cdots c (z_1)] |_{\tmscript{\begin{array}{l}
       \text{$z_1 = q_{\omega}^{\natural} (s_1)$}\\
       \text{$\cdots$}\\
       \text{$z_n = q_{\omega}^{\natural} (s_n)$}
     \end{array}}} = \lb{2} \lp{2} \sum_{j, k = 1}^n B_{j, k} \partial_{z_j}
     \cdot \partial_{z_k} \rp{2} c (z_n) \cdots c (z_1) \rb{2}
     \ve{2}_{\tmscript{\begin{array}{l}
       \text{$z_1 = q_{\omega}^{\natural} (s_1)$}\\
       \text{$\cdots$}\\
       \text{$z_n = q_{\omega}^{\natural} (s_n)$}
     \end{array}}} . \]
  Then $\partial^2_x F_n = G_n$. Hence
  \[ \text{$(\tmop{interior}) = \partial^2_x v_n$} . \]
  This proves that the function $u : = u_{\omega} v$ is solution of
  (\ref{barbara20}).
  
  In the case $\omega = 0$, \ $\Omega . \xi \otimes_n \xi = t \bar{\Omega}
  \cdot \xi \otimes_n \xi$ and $q_{\omega} (s) = y + \frac{s}{t} (x - y)$. We
  can get directly that $v \in \mathcal{A} ( \demiplan{} \times \mathbbm{C}^{2
  \nu})$ and that $u = u_0 v$ is solution of (\ref{barbara20}).
\end{proof}

\begin{remark}
  Weaker assumptions than (\ref{anna1}) can also be considered. One can assume
  that, for every $m \geqslant 0$,
  \begin{equation}
    \label{anna12} \int_{\mathbbm{R}^{\nu}} | \xi |^m d| \mu | (\xi) < \infty
    .
  \end{equation}
  Then, for instance in the free case, $v \in \mathcal{C}^{\infty} \lp{1}
  \overline{\demiplan{}} \times \mathbbm{R}^{2 \nu} \rp{1} \cap \mathcal{A}
  \lp{1} \demiplan{}, \mathcal{C}^{\infty} (\mathbbm{R}^{2 \nu}) \rp{1}$.
\end{remark}

\subsection{Proof of Theorem \ref{theorembarbara1} and Corollary
\ref{corollairebarbara3}}

In the study of Borel transform of the conjugate heat kernel in the free case,
the Borel transform of $t \longmapsto t^n \exp (- Bt)$, $B \in \mathbbm{C}$,
which is related to Bessel functions, plays a specific role.

\begin{lemma}
  For every \ $z \in \mathbbm{C}$, let us denote
  \[ J (z) \assign \sum_{n \geqslant 0} (- 1)^n \frac{z^n}{(n!)^2} =
     \int_0^{\pi} \cos \lp{1} 2 z^{1 / 2} \sin (\varphi) \rp{1} \frac{d
     \varphi}{\pi} . \]
  The following estimate holds for every \ $z \in \mathbbm{C}$
  \begin{equation}
    \label{carmina14} |J (z) | \leqslant \exp \lp{1} 2| \mathcal{I} mz^{1 / 2}
    | \rp{1} .
  \end{equation}
\end{lemma}

These properties are easy to prove and well known since $J_{} (z) = J_0 \lp{1}
2 z^{1 / 2} \rp{1}, J_0$ denoting the Bessel function of order $0$.

\begin{lemma}
  \label{lemmacarmina2}Let $n \geqslant 1$. For $B \in \mathbbm{C}$, let $\tau
  \longmapsto K_n (B, \tau)$ be the Borel transform of the function $t
  \longmapsto t^n \exp (- Bt)$. Then, for $\tau \in \mathbbm{C}$,
  \begin{equation}
    \label{carmina16} K_n (B, \tau) = \tau^n \sum_{m \geqslant 0} \frac{(-
    1)^m}{m! (m + n) !} (B \tau)^m = \tau^n \int_0^1 \frac{(1 - \theta)^{n -
    1}}{(n - 1) !} J (\theta B \tau) d \theta .
  \end{equation}
  \begin{equation}
    \label{carmina18} |K_n (B, \tau) | \leqslant \frac{| \tau |^n}{n!} \exp
    \lp{2} 2 \sqrt{B} | \mathcal{I} m (\tau^{1 / 2}) | \rp{2} \tmop{for} B
    \geqslant 0 .
  \end{equation}
\end{lemma}

\begin{proof}
  Since $t^n \exp (- Bt) = \sum_{m \geqslant 0} (- 1)^m \frac{t^{n + m}}{m!}
  B^m$, using the definition of the Borel transform, we get the first equality
  of (\ref{carmina16}). Since
  \[ J (\theta B \tau) = \sum_{m \geqslant 0} (- 1)^m \frac{(\theta B
     \tau)^m}{(m!)^2}, \]
  the second equality of (\ref{carmina16}) is a consequence of $\int_0^1
  \frac{(1 - \theta)^{n - 1}}{(n - 1) !} \frac{\theta^m}{m!} d \theta =
  \frac{1}{(m + n) !}$. (\ref{carmina18}) is a consequence of
  (\ref{carmina14}) and (\ref{carmina16}).
\end{proof}

Now we prove Theorem \ref{theorembarbara1}. By Proposition (\ref{annaprop1})
with $\omega = 0$, the function $u : = u_0 v$ is solution of (\ref{barbara10})
where $v \assign \mathbbm{1} + \sum_{n \geqslant 1} v_n$,
\begin{equation}
  \label{carmina21} v_n (t, x, y) \assign \int_{0 < s_1 < \cdots < s_n < 1}
  \int_{\mathbbm{R}^{\nu n}} G_n d^{\nu n} \mu^{\otimes} (\xi) d^n s,
\end{equation}
and
\[ G_n \assign t^n \exp (- t \bar{\Omega} \cdot \xi \otimes_n \xi) \exp
   (iq_0^n (s) \cdot \xi) . \]
Here $q_0^n (s) \cdot \xi \assign q_0 (s_1) \cdot \xi_1 + \cdots + q_0 (s_n)
\cdot \xi_n$ where $q_0 (s) \assign y + s (x - y)$. Since $q_0^n (s) \cdot
\xi$ does not depend on $t$, Lemma \ref{lemmacarmina2} will allow us to obtain
a convenient formulation of the Borel transform of $v$.

Let $\hat{F}_n, \hat{v}_n, \hat{v}$ be defined by
\begin{equation}
  \label{carmina22} \hat{F}_n \assign \exp (iq_0^n (s) \cdot \xi) K_n (
  \bar{\Omega} \cdot \xi \otimes_n \xi, \tau),
\end{equation}
\begin{equation}
  \label{carmina24} \hat{v}_n (\tau, x, y) \assign \int_{0 < s_1 < \cdots <
  s_n < 1} \int_{\mathbbm{R}^{\nu n}} \hat{F}_n d^{\nu n} \mu^{\otimes} (\xi)
  d^n s,
\end{equation}
\begin{equation}
  \label{carmina25} \text{$\hat{v} \assign \mathbbm{1} + \sum_{n \geqslant 1}
  \hat{v}_n$} .
\end{equation}
Let $\kappa, R > 0$. Let $x, y \in \mathbbm{C}^{\nu}$ be such that $|
\mathcal{I} mx|, | \mathcal{I} my| < R$ and let $\tau \in \parabole{\kappa}$.
Let us check that the integral giving the definition of $\hat{v}_n$ is
absolutely convergent and hence that the series giving the definition of
$\hat{v}_{}$ is absolutely convergent. By (\ref{carmina18})
\[ \text{$| \hat{F}_n | \leqslant \exp \lp{1} R (| \xi_1 | + \cdots + | \xi_n
   |_{}) \rp{1} \times \frac{| \tau |^n}{n!} \exp \lp{2} 2 \sqrt{\bar{\Omega}
   \cdot \xi \otimes_n \xi} | \mathcal{I} m (\tau^{1 / 2}) | \rp{2}$} . \]
By Lemma \ref{lemmaclaudia1}
\begin{eqnarray*}
  2| \mathcal{I} m (\tau^{1 / 2}) | \sqrt{\bar{\Omega} \cdot \xi \otimes_n
  \xi} \leqslant &  & 2 \kappa^{1 / 2} \sqrt{n (\xi_1^2 + \cdots + \xi_n^2)}\\
  \leqslant &  & 2 \times \lp{2} \frac{2 \kappa n}{\varepsilon} \srp{2}{1 /
  2}{} \times \sqrt{\frac{\varepsilon}{2} (\xi_1^2 + \cdots + \xi_n^2)}\\
  \leqslant &  & \frac{2 \kappa n}{\varepsilon} + \frac{\varepsilon}{2}
  (\xi_1^2 + \cdots + \xi_n^2) .
\end{eqnarray*}
Hence
\begin{equation}
  \label{carmina26} \text{$| \hat{F}_n | \leqslant \exp \lp{1} R (| \xi_1 | +
  \cdots + | \xi_n |_{}) \rp{1} \times \frac{| \tau |^n}{n!}$} \exp \lp{1}
  \frac{2 \kappa n}{\varepsilon} \rp{1} \exp \lp{2} \frac{\varepsilon}{2}
  (\xi_1^2 + \cdots + \xi_n^2) \rp{2} .
\end{equation}
Let
\[ A \assign \int \exp \lp{1} \frac{2 \kappa}{\varepsilon} +
   \frac{\varepsilon}{2} \xi^2 + R| \xi | \rp{1} d| \mu | (\xi) . \]
The integral in (\ref{carmina24}) is absolutely convergent and
\[ | \hat{v}_n | \leqslant \frac{1}{(n!)^2} \times (A | \tau |)^n . \]
Then the series in (\ref{carmina25}) is absolutely convergent. Since $|
\hat{v} | \leqslant \sum_{n \geqslant 0} \frac{1}{(n!)^2} \times (A | \tau
|)^n = J (- A | \tau |)$ and by (\ref{carmina14})
\[ | \hat{v} | \leqslant \exp \lp{1} 2 \lp{0} A | \tau | \srp{0}{1 / 2}{}
   \rp{1} . \]
This proves (\ref{barbara12}).

By dominated convergence theorem, one can check that $\hat{v}_n$ and hence
$\hat{v}$ are analytic on $\parabole{\kappa} \times \mathbbm{C}^{2 \nu}$ and
hence on $\mathbbm{C}^{1 + 2 \nu}$ since $\kappa$ is arbitrary.

Let us prove that $v$ is the Laplace transform of $\hat{v}$. Let $t \in] 0, +
\infty [$. By (\ref{carmina22}) and \ since $t \longmapsto t^n \exp (- Bt)$ is
the Laplace transform of $\tau \longmapsto K_n (B, \tau)$,
\[ t^n \exp (iq_0^n (s) \cdot \xi) \exp (- t \bar{\Omega} \cdot \xi \otimes_n
   \xi) = \int_0^{+ \infty} \hat{F}_n (\tau) e^{- \frac{\tau}{t}} \frac{d
   \tau}{t} . \]
Then, by (\ref{carmina21}), $v$ is the Laplace transform of $\hat{v}$.

Now we may prove Corollary \ref{corollairebarbara3}. Let $p_t (x, y)$ be the
solution of
\begin{equation}
  \label{carmina30} \begin{array}{l}
    \left\{ \begin{array}{l}
      \partial_t p_t (x, y_{}) = \lp{1} \partial^2_x + c (x) \rp{1} p_t (x,
      y)\\
      \\
      p|_{t = 0} (x, y) =_{} \sum_{q \in \mathbbm{Z}^{\nu}} \delta_{x = y + q}
      \mathbbm{1}
    \end{array} \right.
  \end{array}
\end{equation}
Since the torus is compact and $c$ is a $\tmop{Hermitian}
\tmop{matrix}$-valued analytic function, the spectrum of $H$ is real, discret
and, for $t \in] 0, + \infty [$,
\[ \sum_{n = 1}^{+ \infty} e^{- \lambda_n t} = \int_{[0, 1]^{\nu}} \tmop{Tr}
   \lp{1} p_t (x, x) \rp{1} dx. \]
Let $\mu \assign \sum_{q \in \mathbbm{Z}^{\nu}} c_q \delta_{\xi = 2 \pi q}$.
Then $c$ satisfies the assumptions of Theorem \ref{theorembarbara1}. Let
$\hat{v}$ and $C$ be defined as in Theorem \ref{theorembarbara1}. Then
\begin{equation}
  \label{carmina31} p_t (x, y) = (4 \pi t)^{- \nu / 2} \sum_{q \in
  \mathbbm{Z}^{\nu}} \exp \lp{2} - \frac{(x - y - q)^2}{4 t} \rp{2} \int_0^{+
  \infty} \text{$\hat{v}$} (\tau, x, y + q) e^{- \frac{\tau}{t}} \frac{d
  \tau}{t} .
\end{equation}
Let $\kappa > 0$. By (\ref{barbara12}), the quantity $\text{$\hat{v}$} (\tau,
x, y + q)$ is bounded (uniformly in $q \in \mathbbm{Z}^{\nu}$) by $\exp \lp{1}
C| \tau |^{1 / 2} \rp{1}$ and $C$ only depends on $\max \lp{1} | \im x|, | \im
y| \rp{1}$ and $\kappa$. Hence the series in (\ref{carmina31}) is convergent
for $t \in \demiplan{}$, $x \in \mathbbm{C}^{\nu}$ and $y \in
\mathbbm{C}^{\nu}$. Let
\[ \hat{w} (q, \tau) \assign \int_{[0, 1]^{\nu}} \tmop{Tr} \lp{1}
   \text{$\hat{v}$} (\tau, x, x + q) \rp{1} dx, \]
\[ C_1 \assign \sup \lbc{2}_{} \frac{| \tmop{Tr} M|}{|M|} \ve{1} M \neq 0
   \rbc{2}, C_2 \assign 2 \lp{2} \int \exp \lp{1} \frac{2 \kappa}{\varepsilon}
   + \frac{\varepsilon}{2} \xi^2 \rp{1} d| \mu | (\xi) \srp{2}{1 / 2}{} . \]
Here the supremum is taken over non-vanishing $d \times d$ complex matrices.
Then (\ref{barbara16a}) and (\ref{barbara16b}) hold.

\subsection{Proof of Theorem \ref{theorembarbara3}}

The following lemma will be useful for the proof of Theorem
\ref{theorembarbara3}.

\begin{lemma}
  \label{sonialemme2}Let $T > 0$, $K_1 > 0,$and $\sigma_1 > 0$. Let $f$ be an
  analytic function on $\demidisque{T}$. Assume that, for every $r \geqslant
  0$, there exist $f_r$ analytic on $\disque{T}$ and $g_r$ analytic on
  $\demidisque{T}$ such that
  \[ f = f_r + g_r, \]
  \[ |f_r (t) | \leqslant K_1 \sigma_1^r r! \text{for every } t \in
     \disque{T}, \]
  \[ |g_r (t) | \leqslant K_1 \sigma_1^r r!|t|^r  \text{for every } t \in
     \demidisque{T} . \]
  Let $\sigma_2 > \sigma_1$. Then there exist $K_2 > 0$, $a_0, a_1, \ldots \in
  \mathbbm{C}$, $R_0, R_1, \ldots \in \text{$\mathcal{A} ( \demidisque{T})$}$
  such that
  \begin{equation}
    \label{annabelle3} \left\{ \begin{array}{l}
      f (t) = a_0 + \cdots + a_{r - 1} t^{r - 1} + R_r (t)\\
      \\
      |R_r (t) | \leqslant K_2 \sigma_2^r r!|t|^r
    \end{array}, \right.
  \end{equation}
  for every $r \geqslant 0$ \ and $t \in \demidisque{T}$.
\end{lemma}

\begin{proof}
  Let $r \geqslant 0$. Let  M:=$\sup_{|t| < T} |f_r (t) |$ and $a_0 \assign
  f_r (0), \ldots, a_{r - 1} \assign \frac{1}{(r - 1) !} f_r^{(r - 1)} (0)$. \
  Let $w_r$ be the analytic function on $\disque{T}$ defined by
  \[ w_r (t) = \frac{1}{t^r} \left( f_r (t) - (a_0 + \cdots + a_{r - 1} t^{r -
     1}) \right) . \]
  Let $q = 0, \ldots, r - 1$ and let $t \in \text{$\disque{T}$}$. By Cauchy
  formula, $|a_q t^q | \leqslant M$. By maximum modulus principle
  \[ |w_r (t) | \leqslant \frac{M}{T^r} (1 + r) \leqslant \frac{\text{$K_1
     \sigma_1^r r!$}}{T^r} (1 + r) . \]
  Let $\sigma_2 > \sigma_1$. Let us choose $K_2 > 0$ such that, for $r
  \geqslant 0, \frac{\text{$K_1 \sigma_1^r$}}{T^r} (1 + r) \leqslant
  \frac{1}{2} K_2 \sigma_2^r$ and $\text{$K_1 \sigma_1^r \leqslant
  \frac{1}{2}$} K_2 \sigma_2^r$. Let $R_r (t) \assign t^r w_r (t) + g_r (t)$.
  Then (\ref{annabelle3}) holds for $t \in \demidisque{T}$.
\end{proof}

\begin{remark}
  \label{soniaremark2}In fact, we use a parametric version of Lemma
  \ref{sonialemme2}. The function $f$ may depend on additional parameters. Of
  course, the constants which appear in the assumptions of Lemma
  \ref{sonialemme2} must be independent of the parameters. Moreover, we have
  to assume that functions take their values in a complex finite dimensional
  space.
\end{remark}

Now we prove Theorem \ref{theorembarbara3}. Instead of proving
(\ref{barbara21}) and (\ref{barbara22}) directly, we shall prove the existence
of the following factorization
\begin{equation}
  \label{annabelle6} u = u_{\omega} v_{\omega}
\end{equation}
where
\[ v_{\omega} \assign \mathbbm{1} + a^w_1 (x, y) t + \cdots + a^{\omega}_{r -
   1} (x, y) t^{r - 1} + R^{\omega}_r (t, x, y) \]
and $R^{\omega}_r$ satisfying (\ref{barbara22}). By (\ref{anna0.0}),
$\frac{u_{\omega}}{u_0}$ is analytic near $t = 0$. Then (\ref{barbara21}) and
(\ref{barbara22}) will hold since Borel summability properties of an expansion
do not change if it is multiplied by a convergent expansion near $t = 0$.

Let $v$ and $T_c$ be respectively defined by (\ref{anna3}) and (\ref{anna4}).
Then $v$ verifies all the properties of Proposition \ref{annaprop1}. We shall
use Lemma \ref{sonialemme2} with Remark \ref{soniaremark2}. For $m \geqslant
1, w \in \mathbbm{C}$,
\begin{equation}
  \label{annabelle7} \text{$e^w = 1 + w + \cdots + \frac{w^{m - 1}}{(m - 1) !}
  + \frac{w^m}{(m - 1) !} \int_0^1 (1 - \theta)^{m - 1} e^{\theta w} d
  \nonesep \theta^{}$} .
\end{equation}
Let $r > 0$. Using (\ref{annabelle7}) with $w = - \Omega . \xi \otimes_n \xi$
in (\ref{anna3}), one gets, for $t \in \demidisque{T_c}$, a decomposition $v =
f_r + g_r$ where
\[ \text{$f_r \assign \mathbbm{1} + \sum_{\tmscript{\begin{array}{l}
     n + m < r\\
     n \geqslant 1, m \geqslant 0
   \end{array}}} t^n \int_{0 < s_1 < \cdots < s_n < 1} \int_{\mathbbm{R}^{\nu
   n}} F_{n, m} d^{\nu n} \mu^{\otimes} (\xi) d^n s,$} \]
\[ F_{n, m} \assign \frac{1}{m!} (- \Omega . \xi \otimes_n \xi)^m \exp \lp{1}
   iq_{\omega}^n (s) \cdot \xi) \rp{1}, \]
and $g_r = g_{r^{}}^{(1)} + g_{r^{}}^{(2)}$, where
\[ \text{$g_{r^{}}^{(1)} \assign \sum_{\tmscript{\begin{array}{l}
     n + m = r\\
     n \geqslant 1, m \geqslant 1
   \end{array}}} t^n \int_{0 < s_1 < \cdots < s_n < 1} \int_{\mathbbm{R}^{\nu
   n}} \int_0^1 G_{n, m} d \nonesep \theta d^{\nu n} \mu^{\otimes} (\xi) d^n
   s$,} \]
\[ G_{n, m} \assign \frac{1}{(m - 1) !} (- \Omega . \xi \otimes_n \xi)^m (1 -
   \theta)^{m - 1} \exp (- \theta \Omega . \xi \otimes_n \xi) \exp
   (iq_{\omega}^n (s) \cdot \xi), \]
and
\[ g_{r^{}}^{(2)} \assign \sum_{n \geqslant r} v_n . \]
In what follows, we shall fix some $R > 0$ and some $\varepsilon \in] 0,
\varepsilon_0 [$. We denote
\[ A \assign \int_{\mathbbm{R}^{\nu}} \exp \lp{1} \varepsilon \xi^2 + 4 R| \xi
   | \rp{1} d| \mu | (\xi) \]
and we take arbitrary $x, y \in \mathbbm{C}^{\nu}$ such that $|x| < R$ \ and
$|y| < R$.

We begin to check that $f_r$ verifies assumptions of Lemma \ref{sonialemme2}.
Obviously, by definition of $\Omega . \xi \otimes_n \xi$ \ and $q_{\omega}^n
(s) \cdot \xi$, $f_r$ are analytic for $t \in \disque{\frac{\pi}{| \omega
|}}$. By (\ref{claudia30}), (\ref{claudia32}) and (\ref{anna04.a}), for $t \in
\disque{T_c}$,
\[ |f_r | \leqslant 1 + \sum_{\tmscript{\begin{array}{l}
     n + m \leqslant r - 1\\
     n \geqslant 1, m \geqslant 0
   \end{array}}} \int_{0 < s_1 < \cdots < s_n < 1} \int_{\mathbbm{R}^{\nu n}}
   H_{n, m} d^{\nu n} | \mu |^{\otimes} (\xi_{}) d^n s_{} \]
where
\[ H_{n, m} \assign \frac{T_c^n}{m!} \lp{1} 2 nT_c ( \sum_{j = 1}^n \xi_j^2)
   \rp{1}^m \exp \lp{1} 4 R (| \xi_1 | + \cdots | \xi_n |) \rp{1} . \]
But
\begin{equation}
  \label{annabelle10} \frac{\varepsilon^m}{m!} \lp{1} \sum_{j = 1}^n \xi_j^2
  \rp{1}^m \leqslant \exp \lp{1} \varepsilon \sum_{j = 1}^n \xi_j^2 \rp{1}
  \leqslant \exp (\varepsilon \xi_1^2) \cdots \exp (\varepsilon \xi_n^2) .
\end{equation}
Then
\[ |f_r | \leqslant 1 + \sum_{\tmscript{\begin{array}{l}
     n + m \leqslant r - 1\\
     n \geqslant 1, m \geqslant 0
   \end{array}}} \frac{n^m}{n!} \left( \frac{2}{\varepsilon} \right)^m A^n
   T_c^{m + n} . \]
Let $B \assign \max \lp{1} 1, \frac{2}{\varepsilon}, A, T_c \rp{1}$. Then, for
$(n, m) \in \mathbbm{N}^2$ such that $n + m < r, n \geqslant 1, m \geqslant
0$,
\[ \frac{n^m}{n!} \left( \frac{2}{\varepsilon} \right)^m A^n T_c^{m + n}
   \leqslant r^r B^{2 r} . \]
Since $\sum_{\tmscript{\begin{array}{l}
  n + m < r\\
  n \geqslant 1, m \geqslant 0
\end{array}}} 1 = \frac{r (r - 1)}{2}$,
\[ |f_r | \leqslant 1 + \frac{r (r - 1)}{2} r^r B^{2 r} . \]
Then Stirling formula implies the existence of $K_1 > 0$ \ and $\sigma_1 > 0$
such that for $t \in \disque{T_c}$ and $r \geqslant 1,$
\begin{equation}
  \label{annabelle12} |f_r | \leqslant K_1 \sigma_1^r r!.
\end{equation}
We check now that $g_r$ verifies the assumptions of Lemma \ref{sonialemme2}.
Let $t \in \demidisque{T_c}$. By (\ref{claudia30}), (\ref{claudia32}) and
(\ref{anna04.a})
\[ |g_r^{(1)} | \leqslant \sum_{\tmscript{\begin{array}{l}
     n + m = r\\
     n \geqslant 1, m \geqslant 1
   \end{array}}} \int_{0 < s_1 < \cdots < s_n < 1} \int_{\mathbbm{R}^{\nu n}}
   L_{n, m} d^{\nu n} | \mu |^{\otimes} (\xi_{}) d^n s_{} \]
where
\[ L_{n, m} \assign \frac{|t|^n}{m!} \lp{1} 2 n|t| ( \sum_{j = 1}^n \xi_j^2)
   \rp{1}^m \exp \lp{1} 4 R (| \xi_1 | + \cdots | \xi_n |) \rp{1} . \]
By (\ref{annabelle10})
\[ |g_r^{(1)} | \leqslant |t|^r \sum_{\tmscript{\begin{array}{l}
     n + m = r\\
     n \geqslant 1, m \geqslant 1
   \end{array}}} \frac{n^m}{n!} \left( \frac{2}{\varepsilon} \right)^m A^n .
\]
By (\ref{anna4.5}), $|v_n | \leqslant \frac{1}{n!} A^n |t|^n$. Hence
\[ |g_{r^{}}^{(2)} | \leqslant \frac{|t|^r}{T_c^r} \exp (AT_c) . \]
Since $g_{r^{}} = g_r^{(1)} + g_{r^{}}^{(2)}$and by similar arguments which
allow us to obtain (\ref{annabelle12}), we get the existence of $K_1 > 0$ \
and $\sigma_1 > 0$ such that for $t \in \demidisque{T_c}$ and $r \geqslant 1,$
\begin{equation}
  \label{annabelle14} |g_r (t) | \leqslant K_1 \sigma_1^r r!|t|^r .
\end{equation}
By (\ref{annabelle12}) and (\ref{annabelle14}), we can use Lemma
\ref{sonialemme2} with Remark \ref{soniaremark2}. Then factorization
(\ref{annabelle6}) holds and $R^{\omega}_r$ satisfies (\ref{barbara22}). This
proves (\ref{barbara21}) and (\ref{barbara22}).

Let us prove the regularity of the functions $a_1, a_2, \ldots$ and $R_0, R_1,
\ldots$ The function $V$ defined by $V (x) = - \frac{\omega^2}{4} x^2 + c (x)$
is analytic on $\mathbbm{C}^{\nu}$. By a property of Minakshisundaram-Pleijel
expansion , $a_1, a_2, \ldots \in \mathcal{A} ( \mathbbm{C}^{2 \nu})$ and
hence $R_0, R_1, \ldots \in \mathcal{A} ( \demidisque{T_c} \times
\mathbbm{C}^{2 \nu})$.

\section{Appendix}

Let $u$ be the solution of (\ref{barbara20}) and $v$ be the function defined
by the factorization $u = u_{\omega} v.$ The goal of this section is to give a
heuristic explanation of the following (deformation) formula
\begin{equation}
  \label{maria3} v = \mathbbm{1} + \sum_{n \geq 1} v_n,
\end{equation}
where
\[ v_n \assign \int_{0 < s_1 < \cdots < s_n < t} F_n ds_1 \cdots ds_n, \]
\[ F_n \assign \lb{2} \exp ( \sum_{j, k = 1}^n \Omega_{j, k}^{\natural}
   \partial_{z_j} \cdot \partial_{z_k}) c (z_n) \cdots c (z_1) \rb{2}
   \ve{2}_{\tmscript{\begin{array}{l}
     \text{$z_1 = q_{\omega}^{\natural} (s_1)$}\\
     \text{$\cdots$}\\
     \text{$z_n = q_{\omega}^{\natural} (s_n)$}
   \end{array}}} . \]

Here $t > 0$. $\Omega^{\natural}$ and $q_{\omega}^{\natural}$ are defined in
section \ref{proof_theorem} (cf. (\ref{claudia0}) and (\ref{anna03})). We
prove and use this formula in a rigorous context (cf. (\ref{anna6})). This
heuristic viewpoint has the advantage to explain the shape of the formula and
gives a simple interpretation of matrix $\Omega^{\natural}$. However, from a
rigorous point of view, the proof chosen in Proposition \ref{annaprop1} seems
to be more efficient.

For proving (\ref{maria3}), we use Wiener integral representation of $u$ and
the following version of Wick's theorem. Let $E$ be a real vector space. Let
$\left\langle ., . \right\rangle$ be a scalar product on $E$ and let $A$ be a
symmetric invertible operator. Then
\begin{equation}
  \label{maria5} \frac{\int_E \exp \lp{1} \left\langle \theta, x \right\rangle
  \rp{1} \exp \lp{1} - \frac{1}{4} \left\langle Ax, x \right\rangle \rp{1}
  dx}{\int_E \exp \lp{1} - \frac{1}{4} \left\langle Ax, x \right\rangle \rp{1}
  dx} = \exp \lp{1} \left\langle A^{- 1} \theta, \theta \right\rangle \rp{1} .
\end{equation}
Of course, if $E$ has finite dimension and $A > 0$, this formula is rigorous
and means that the Laplace transform of a Gaussian is a Gaussian.

Let $V (x) = - \frac{\omega^2}{4} x^2 + c (x)$. The justification of the
deformation formula starts from the following representation of $u$
\begin{equation}
  \label{maria6} u = \int \tmop{Texp} \lp{2} \int_{0^{}}^t V \lp{1} q (s)
  \rp{1} \tmop{ds} \rp{2} \exp \lp{2} - \frac{1}{4} \int_0^t \dot{q}^2 (s) ds
  \rp{2} dq,
\end{equation}
where integration is taken over all the paths $q : \mathbbm{R} \longrightarrow
\mathbbm{R}^{\nu}$ such that $q (0) = y$ and $q (t) = x$. $\tmop{Texp}$
denotes the time-ordered exponential function:
\[ \tmop{Texp} \lp{2} \int_{0^{}}^t V \lp{1} q (s) \rp{1} \tmop{ds} \rp{2}
   \assign \sum_{n \geqslant 0} \int_{0 < s_1 < \cdots < s_n < t} V \lp{1} q
   (s_n) \rp{1} \cdots V \lp{1} q (s_1) \rp{1} ds_1 \cdots ds_n . \]
Then
\[ u = \int \tmop{Texp} \lp{2} \int_{0^{}}^t c \lp{1} q (s) \rp{1} \tmop{ds}
   \rp{2} \exp \lp{2} - \frac{1}{4} \int_0^t \lp{1} \dot{q}^2 (s) + \omega^2
   q^2 (s) \rp{1} ds \rp{2} dq \]
Using the decomposition $q (s) = q^{\natural}_{\omega} (s) + w (s)$ and the
fact that $q^{\natural}_{\omega}$ is an extremal path for the action
\[ S (q) \assign \int_0^t \lp{1} \dot{q}^2 (s) + \omega^2 q^2 (s) \rp{1} ds,
\]
which is quadratic with respect to $q$, we get $S (q) = S
(q_{\omega}^{\natural}) + S (w)$. Hence
\begin{equation}
  \label{maria7} u = \exp \lp{1} - \frac{1}{4} S (q^{\natural}_{\omega})
  \rp{1} \int \tmop{Texp} \lp{2} \int_{0^{}}^t c \lp{1} q_{\omega}^{\natural}
  (s) + w (s) \rp{1} \tmop{ds} \rp{2} \exp \lp{1} - \frac{1}{4} S (w) \rp{1}
  dw,
\end{equation}
where integration is taken over all the paths $w : \mathbbm{R} \longrightarrow
\mathbbm{R}^{\nu}$ such that $w (0) = 0$ and $w (t) = 0$. By (\ref{maria7}),
with $c \equiv 0$,
\[ u_{\omega} = \exp \lp{1} - \frac{1}{4} S (q^{\natural}_{\omega}) \rp{1}
   \int \exp \lp{1} - \frac{1}{4} S (w) \rp{1} dw. \]
Then $u = u_{\omega} v$ where
\[ v = \frac{\int \tmop{Texp} \lp{2} \int_{0^{}}^t c \lp{1}
   q^{\natural}_{\omega} (s) + w (s) \rp{1} \tmop{ds} \rp{2} \exp \lp{1} -
   \frac{1}{4} S (w) \rp{1} dw}{\int \exp \lp{1} - \frac{1}{4} S (w) \rp{1}
   dw} . \]
Then $v = \sum_{n \geqslant 0} \int_{0 < s_1 < \cdots < s_n < t} \tilde{v}_n
ds_1 \cdots ds_n$ where
\begin{equation}
  \label{maria8} \tilde{v}_n \assign \frac{\int_{} c \lp{1}
  q^{\natural}_{\omega} (s_n) + w (s_n) \rp{1} \cdots c \lp{1}
  q^{\natural}_{\omega} (s_1) + w (s_1) \rp{1} \exp \lp{1} - \frac{1}{4} S (w)
  \rp{1} dw}{\int \exp \lp{1} - \frac{1}{4} S (w) \rp{1} dw}
\end{equation}
By Taylor formula
\[ c \lp{1} q^{\natural}_{\omega} (s) + w (s) \rp{1} = [\exp (w (s) \cdot
   \partial_z) c (z)] |_{z = q^{\natural}_{\omega} (s)} . \]
Then
\begin{equation}
  \label{maria9} \tilde{v}_n = \lb{3} \frac{\int_{} \exp \lp{1} \sum_{j = 1}^n
  w (s_j) \cdot \partial_{z_j} \rp{1} \exp \lp{1} - \frac{1}{4} S (w) \rp{1}
  dw}{\int \exp \lp{1} - \frac{1}{4} S (w) \rp{1} dw} c (z_n) \cdots c (z_1)
  \rb{3} \ve{3}_{\tmscript{\begin{array}{l}
    \text{$z_1 = q^{\natural}_{\omega} (s_1)$}\\
    \text{$\cdots$}\\
    \text{$z_n = q^{\natural}_{\omega} (s_n)$}
  \end{array}}} .
\end{equation}
To use (\ref{maria5}), let us define
\[ E \assign \{w : \mathbbm{R} \longrightarrow \mathbbm{R}^{\nu}, w (0) = w
   (t) = 0\}, \]
\[ \left\langle w_1, w_2 \right\rangle \assign \int_0^t w_1 (s) \cdot w_2 (s)
   ds, \]
\[ A \assign - \frac{d^2}{ds^2} + \omega^2 . \]
Then
\begin{equation}
  \label{maria10} A^{- 1} w (s) = \int_0^t \frac{\tmop{sh} (\omega s \wedge
  s') \tmop{sh} \lp{1} \omega (t - s \vee s') \rp{1}}{\omega \tmop{sh} (\omega
  t)} w (s') ds' .
\end{equation}
As in quantum field theory, the operator $A^{- 1}$ can be viewed as a
propagator. Since
\[ \text{$\sum_{j = 1}^n w (s_j) \cdot \partial_{z_j} = \lba{1} \sum_{j =
   1}^n \delta_{s_j} \partial_{z_j}, w \rba{1}$} \]
and as
\[ \text{$\lba{1} A^{- 1} \sum_{j = 1}^n \delta_{s_j} \partial_{z_j}, \sum_{j
   = 1}^n \delta_{s_j} \partial_{z_j} \rba{1} = \sum_{j, k = 1}^n \Omega_{j,
   k}^{\natural} \partial_{z_j} \cdot \partial_{z_k}$}, \]
(\ref{maria5}) gives
\begin{equation}
  \frac{\int_{} \exp \lp{1} \sum_{j = 1}^n w (s_j) \cdot \partial_{z_j} \rp{1}
  \exp \lp{1} - \frac{1}{4} S (w) \rp{1} dw}{\label{maria11} \int \exp \lp{1}
  - \frac{1}{4} S (w) \rp{1} dw} = \exp \lp{2} \sum_{j, k = 1}^n \Omega_{j,
  k}^{\natural} \partial_{z_j} \cdot \partial_{z_k} \rp{2} .
\end{equation}
The deformation formula (\ref{maria3}) is then a consequence of (\ref{maria9})
and (\ref{maria11}).

\bigskip

REFERENCES
\medskip

[A-H] S. A. Albeverio, R. J.Hoegh-Krohn, \tmtextit{\tmtextup{Mathematical
Theory of Feynman path integrals}}, Lecture Notes in Mathematics
\tmtextbf{523} (1976).

[B-B] R. Balian and C. Bloch, \tmtextit{\tmtextup{Solutions of the
Schr\"odinger equation in terms of classical paths}}, Ann. of Phys.
\tmtextbf{85} (1974), 514-545.

[Co1] Y. Colin de Verdi\`ere, \tmtextit{\tmtextup{Spectre du Laplacien et
longueurs des g\'eod\'esiques p\'eriodiques.I}}, Composio Mathematica
\tmtextbf{27-1} (1973), 83-106.

[Co2] Y. Colin de Verdi\`ere, \tmtextit{\tmtextup{Spectre du Laplacien et
longueurs des g\'eod\'esiques p\'eriodiques.II}} , Composio Mathematica
\tmtextbf{27-2} (1973), 159-184.

[D-D-P] E. Delabaere, H. Dillinger, F. Pham,
\tmtextit{\tmtextup{\tmtextmd{Exact semiclassical expansions for
one-dimensional quantum oscillator}}}, J. Math. Phys. \tmtextbf{38(12)}
(1997), 6126-6184.

[Fu-Os-Wi] Y. Fujiwara, T.A. Osborn, S.F.J. Wilk; Wigner-Kirkwood expansions,
Physical Review A, 25-1 (1982), 14-34.

[Ga] K. Gawedzki, Construction of quantum-mechanical dynamics by means of path
integrals in phase space, Reports on Mathematical Physics, vol. 6, No. 3
(1972), 327-342.

[Ge-Ya] I.M. Gel'fand, A.M. Yaglom, Integration in functional spaces and its
applications in quantum physics, Journal of Mathematical Physics, 1-1 (1960),
48-69.

[Ha3] T. Harg\'e, Noyau de la chaleur en dimension quelconque avec potentiel
en temps petit, prepublication Orsay (1998).

[Ha6] T. Harg\'e, A deformation formula for the heat kernel (2013).

[Ha7] T. Harg\'e, Some remarks on the complex heat kernel on
$\mathbbm{C}^{\nu}$ in the scalar potential case (2013).

[Ha8] T. Harg\'e, Borel summation of the semi-classical expansion of the
partition function associated to the Schr\"odinger equation with a one-well
potential (2013).

[It] K. Ito, \tmtextup{Generalized uniform complex measures in the Hilbertian
metric space with their applications to the Feynman integral}, Fifth berkeley
Symp. on Math. Statist. and Prob. \tmtextbf{2} (1967), 145--161.

[On] E. Onofri, On the high-temperature expansion of the density matrix,
American Journal Physics, 46-4 (1978), 379-382.

[Pa] A. Pazy, Semigroups of Linear Operator and Applications to Partial
Differential Equations.

[Ru] W. Rudin, Real and complex analysis, section 6.

[So] A.D. Sokal, \tmtextup{An improvement of Watson's theorem on Borel
summability}, J. Math. Phys \tmtextbf{21(2)} (1980), 261-263.

[V1] A. Voros, The return of the quartic oscillator.\tmtextit{ \tmtextup{The
complex WKB method}}, Annales de l'institut Henri Poincar\'e (A) Physique
th\'eorique \tmtextbf{39-3} (1983), 211-338.

[V2] A. Voros, \tmtextit{\tmtextup{Schr\"odinger equation from}}
$\tmmathbf{O}(h)$ \tmtextit{\tmtextup{to}} $\tmmathbf{O}(h^{\infty})$, Path
Integrals from meV to MEV (1985), 173-195.

\bigskip

D\'epartement de Math\'ematiques, Universit\'e de Cergy-Pontoise, 95302
Cergy-Pontoise, France.

\end{document}